\documentclass[12pt]{article}

\usepackage{graphicx}
\usepackage{amsfonts, amsmath, amsthm, amssymb, mathrsfs}
\usepackage{url}
\usepackage{bm}
\usepackage{setspace}
\usepackage{float}
\usepackage{threeparttable}
\usepackage{multirow}
\usepackage{array}
\usepackage{xcolor}
\usepackage{colortbl}

\usepackage[top=1in, bottom=1in, left=1in, right=1in]{geometry}
\makeatletter
\renewcommand\float@endH{\@endfloatbox\vskip\intextsep
 \if@flstyle\setbox\@currbox\float@makebox\columnwidth\fi
 \box\@currbox\vskip\intextsep\relax\@doendpe}
\makeatother

\usepackage{caption}
\theoremstyle{plain}
\newtheorem{dfn}{Definition}
\newtheorem{prop}{Proposition}
\newtheorem{lem}{Lemma}
\newtheorem{thm}{Theorem}
\newtheorem{cor}{Corollary}
\newtheorem{rem}{Remark}
\newtheorem{ass}{Assumption}

\usepackage[natbib=true, backend=biber, citestyle=authoryear-comp, maxcitenames=2, bibstyle=numeric, giveninits=true, sorting=nyt, uniquename=false, uniquelist=false, 
doi=false,isbn=false,url=false]{biblatex}
\DeclareFieldFormat{pages}{#1}
\DeclareFieldFormat*{title}{#1}
\DeclareFieldFormat*{number}{(#1)}
\DeclareFieldFormat[article]{journaltitle}{\emph{#1}}
\AtBeginBibliography{}
\AtBeginBibliography{\setcounter{maxnames}{99}}
\AtBeginBibliography{\setcounter{minnames}{99}}
\DeclareNameAlias{editor}{author}
\DeclareNameAlias{default}{family-given}

\DeclareBibliographyDriver{article}{
 \printnames{author}
 \setunit{\addspace}
 \printtext{\mkbibparens{\printfield{year}\printfield{extradate}}}
 \setunit{\addperiod\space}
 \printfield{title}
 \setunit{\addcomma\space}
 \printfield{journaltitle}
 \setunit{\addcomma\space}
 \printfield{volume}
 \printfield{number}
 \setunit{\addcomma\space}
 \printfield{pages}
 \finentry
}

\DeclareBibliographyDriver{inproceedings}{
 \printnames{author}
 \setunit{\addspace}
 \printfield[parens]{year}
 \setunit{\addperiod\space}
 \printfield{title}
 \setunit{\addcomma\space}
 \printtext{in}
 \printfield{booktitle}
 \setunit{\addcomma\space}
 \printlist{organization}
 \setunit{\addcomma\space}
 \printfield{pages}
 \finentry
}

\DeclareBibliographyDriver{inbook}{
 \printnames{author}
 \setunit{\addspace}
 \printfield[parens]{year}
 \setunit{\addperiod\space}
 \printfield{title}
 \setunit{\addcomma\space}
 \printtext{in}
 \printnames{editor}
 \printtext{(eds)}
 \printfield{booktitle}
 \setunit{\addcomma\space}
 \printlist{publisher}
 \setunit{\addcomma\space}
 \printfield{pages}
 \finentry
}

\usepackage[colorlinks=true, allcolors=blue]{hyperref}

\title{Worst-case Strategy-proofness\footnote{
The author is deeply indebted to Toyotaka Sakai for many helpful discussions and warm encouragement. 
He also gratefully acknowledges three anonymous referees, Shintaro Aoi, and the participants at the 2025 Autumn Meeting of the Japan Economic Association at Hirosaki University and the 31st Decentralization Conference in Japan at Chuo University for their valuable comments.}}

\author{Kazumasa Ikeda\footnote{Graduate School of Economics, Keio University, Tokyo 108-8345, Japan. Contact: {\tt ikeda-k25@keio.jp}}}

\date{May 27, 2026}

\begin{document}
\maketitle

\begin{abstract}
We introduce a new non-manipulability axiom called \emph{worst-case strategy-proofness (WCSP)}. This axiom is weaker than \emph{strategy-proofness} and stronger than \emph{non-obvious manipulability-worst (NOM-worst)} by Troyan and Morrill (2020). WCSP focuses on non-manipulability in a worst-case scenario. We examine the implications of WCSP in a voting model. Although many voting rules, such as the plurality rule, the Borda rule, and the Dowdall rule, satisfy NOM-worst, they violate WCSP. We obtain a necessary and sufficient condition for the anti-plurality rule with fixed-order tie-breaking to satisfy WCSP in terms of the numbers of agents and alternatives.
\end{abstract}
\bigskip

\smallskip

\newpage

\section{Introduction}
\emph{Strategy-proofness} is a desirable axiom in voting, as it guarantees robustness against any strategic manipulations. Under a voting rule satisfying strategy-proofness, truthful reporting is always optimal regardless of how other agents report. However, the axiom is excessively stringent. As is well known, the Gibbard–Satterthwaite theorem \citep{gibbard1973, satterthwaite1975} demonstrates that any strategy-proof and \emph{onto} voting rule must be \emph{dictatorial}. Consequently, virtually all reasonable voting rules violate strategy-proofness. 
To circumvent this impossibility result, a growing body of recent literature has explored relaxations of strategy-proofness. The present study is aligned with this line of research.
In relaxing strategy-proofness, there remains room for refinement in the notion of strategic manipulation that the axiom presupposes. Strategy-proofness is robust against strategic manipulations by agents who possess perfect information about other agents' reporting; naturally, it also rules out manipulations by less-informed agents. However, in practical voting settings, agents typically do not have access to such perfect information regarding other agents' reporting. In this sense, strategic manipulations considered under strategy-proofness may be based on an excessively demanding informational assumption from a practical perspective. Accordingly, by focusing on manipulations that more closely reflect agents' behavior in practical voting settings, it becomes possible to relax strategy-proofness while still maintaining robustness against practical manipulations.

A pioneering work that relaxes strategy-proofness while incorporating more practical manipulations is a robustness axiom against \emph{obvious manipulations (non-obvious manipulability, NOM)} proposed by \textcite{troyan2020}. They introduce NOM in a general model, including matching and auction models. NOM focuses on manipulations by agents who possess no information about other agents' reporting. A mechanism satisfies NOM if, for each agent, the best-case and worst-case outcomes under truthful reporting are at least as good as the corresponding outcomes under misreporting; these two conditions are referred to as \emph{NOM-best} and \emph{NOM-worst}, respectively. They demonstrate the usefulness of NOM in matching and auction settings.

Subsequently, \textcite{aziz2021obvious} and \textcite{arribillaga2024obvious} apply the notion of NOM to voting models. \textcite{aziz2021obvious} demonstrate that most voting rules, such as the plurality rule, the Borda rule, and the Dowdall rule, satisfy NOM.\footnote{Their results are obtained for cases in which both the number of agents and the number of alternatives are at least three.} The results suggest that the notion of NOM is too weak to capture the idea of non-manipulability in voting models. This is because these voting rules are easily manipulated, and manipulations frequently occur in practical voting settings. In particular, the Borda rule is widely regarded in the literature as highly susceptible to manipulations (see, e.g., \cite{barbie_borda_2006,black_borda_1976,favardin_borda_2002}). Hence, both the notions of strategy-proofness and NOM fail to explain the variation in agents' tendency toward manipulations across voting rules. We consider that there remains room to strengthen NOM by reconsidering the types of manipulation. Unlike in matching and auction settings, in voting settings, the assumption that agents possess no information about other agents' reporting appears excessively weak. In recent real-world voting, such as political elections, public opinion polls or prediction markets provide agents with partial information about other agents' reporting. Accordingly, NOM can be strengthened by considering manipulations by more informed agents.

In light of the above discussion, this study proposes an axiom that is weaker than strategy-proofness yet stronger than NOM(-worst), while capturing a more practically relevant notion of manipulation. We call this new non-manipulability axiom \emph{worst-case strategy-proofness (WCSP)}. WCSP considers agents who do not possess perfect information about other agents' reporting, but are able to form expectations about which alternative is likely to win through information sources such as public opinion polls or prediction markets. 
In particular, we assume that agents who expect that their worst-case outcome will be realized may engage in manipulations in order to avoid it. WCSP guarantees robustness against such worst-case-avoiding manipulations.

Our focus on manipulations intended to avoid the worst-case outcome is motivated by the observation that such manipulations appear salient in practical voting settings. Findings from empirical studies of voting behavior, together with survey data from real-world elections, suggest that worst-case-avoiding manipulations occur frequently. For instance, in experiments on the Borda rule, \textcite{kube2009} show that, for agents with imperfect information about other agents' reporting, manipulations intended to avoid the worst-case outcome occur more frequently than manipulations aimed at bringing about the preferred outcome. Moreover, experimental evidence in \textcite{chytilek2017} indicates that agents' turnout is particularly high when agents seek to prevent the winning of an ``evil'' alternative. In the context of political elections, \textcite{lau1982} further shows that agents often place greater emphasis on negative information about undesirable alternatives and are inclined to report against such alternatives. As these studies suggest, avoiding the worst-case outcome constitutes a general tendency in voting behavior, and worst-case-avoiding manipulations are likely to occur. Our notion of WCSP requires robustness only against such manipulations. In this sense, WCSP imposes a weaker requirement than strategy-proofness, since it focuses only on manipulations associated with a restricted class of worst-case scenarios.

The frequent occurrence of worst-case-avoiding manipulations can be plausibly explained by several factors. First, in such worst-case scenarios, the risk associated with failed manipulation is relatively small. In addition, individuals typically place greater weight on unfavorable outcomes than on favorable ones, as emphasized in behavioral theories such as prospect theory \citep{kahneman1979}. Moreover, prospect theory suggests that individuals tend to exhibit risk-seeking behavior in the domain of losses, and situations involving the avoidance of the worst-case outcome naturally fall into this category. Consistent with this perspective, \textcite{martin2021}, using survey data from real-world elections, shows that agents with stronger risk-seeking tendencies are more likely to engage in manipulations.

In our main results, we show that the anti-plurality rule satisfies WCSP under a certain condition, but most popular scoring rules including the plurality rule, the Borda rule, and the Dowdall rule do not (Theorems \ref{thm:apv}, \ref{thm:pv}, \ref{thm:borda}, and \ref{thm:dowdall}). In fact, we find a necessary and sufficient condition for the anti-plurality rule to satisfy WCSP; the anti-plurality rule with fixed-order tie-breaking satisfies WCSP if and only if the number of agents $n$ is at least $2m-1$, where $m$ is the number of alternatives (Theorem \ref{thm:apv}). Moreover, among the class of scoring rules with fixed-order tie-breaking, no rule other than the anti-plurality rule satisfies WCSP when $n\geq m$ (Theorem \ref{thm:nonapv}). 
Our findings uncover an outstanding non-manipulability feature of the anti-plurality rule,
compared to other scoring rules. We also find a result that overcomes the incompatibility between non-manipulability and ontoness (Proposition \ref{prop:WCSP_onto}).

The paper proceeds as follows. Section 2 presents the model, defines WCSP, and examines relations among axioms. Section 3 presents the main results. Section 4 concludes.

\subsection{Related work}
In recent years, researchers have studied relaxations of strategy-proofness not only in voting but also in a large literature of mechanism design.

A pioneering example is a robustness axiom against obvious manipulations (non-obvious manipulability, NOM) by \textcite{troyan2020}. While strategy-proofness ensures that truthful reporting is a weakly dominant strategy, NOM ensures that it is both a maximin and a maximax strategy; specifically, NOM-worst (resp. NOM-best) corresponds to maximin (resp. maximax) considerations. They show that some non-strategy-proof mechanisms satisfy NOM. For example, the school-proposing DA mechanism and uniform-price auctions satisfy NOM, while the Boston mechanism and pay-as-bid auctions do not. Also, in an assignment model, \textcite{troyan2024rank} shows that rank-minimizing mechanisms with full support satisfy NOM. Furthermore, recent studies examine NOM in various models, such as cake-cutting \citep{ortega2022obvious} and allocation models (\cite{psomas2022}; \cite{arribillaga_bonifacio2025a, arribillaga_bonifacio2025b, arribillaga_risma2025a, arribillaga_risma2025b}).

\textcite{aziz2021obvious} and \textcite{arribillaga2024obvious} examine NOM in voting models. \textcite{aziz2021obvious} show that most voting rules satisfy NOM. Moreover, NOM consists of two robustness conditions, \emph{``best"} and \emph{``worst,"} and they show that, for many voting rules, the worst condition implies the best. \textcite{arribillaga2024obvious} examine NOM for tops-only voting rules, such as the median-voter rule.

We also mention studies that relax strategy-proofness in ways logically independent of NOM and WCSP. \textcite{fernandez2020regret} introduces \emph{regret-free truth-telling} in a matching model. \textcite{arribillaga2022regret} apply this notion to a voting model. \textcite{gori2021manipulation} considers ill-informed agents and proposes \emph{WMG-strategy-proofness} in a voting model. \textcite{dindar2025mini} introduce \emph{minimal strategy-proofness} in a rank aggregation model.

\section{Preliminaries}

\subsection{Model}
Let $N=\{1,\ldots,n\}$ be the set of \emph{agents} with $n\geq2$. Let $X=\{x_1,\ldots,x_m\}$ be the set of \emph{alternatives} with $m\geq2$.
Each agent $i\in N$ has a \emph{(strict) preference} $P_i$ on $X$.\footnote{A binary relation $P_i$ on $X$ is a strict preference if it satisfies \emph{completeness} (for each $x,y\in X$, either $x\mathrel{P_i}y$ or $y\mathrel{P_i}x$), \emph{transitivity} (for each $x,y,z\in X$, $x\mathrel{P_i}y$ and $y\mathrel{P_i}z$ imply $x\mathrel{P_i}z$), and \emph{anti-symmetry} (for no $x,y\in X$, $x\mathrel{P_i}y$ and $y\mathrel{P_i}x$).} Let $\mathcal{P}$ be the set of all preferences on $X$. A \emph{preference profile} is an $n$-tuple ${P}=(P_1,\ldots,P_n) \in \mathcal{P}^n$. As usual, we write $P=(P_i,P_{-i})\in\mathcal{P}^{n}$ to distinguish agent $i$'s preference from the other agents' preferences. A \emph{voting rule} is a function $f:\mathcal{P}^n\to X$ that maps each preference profile $P\in\mathcal{P}^n$ to an alternative $ f(P)\in X$.

Let $R_i$ be the \emph{weak preference} associated with $P_i$; that is, for each $x, y \in X$, $x\mathrel{R_i}y$ if and only if either $x=y$ or $x\mathrel{P_i}y$. 
Let $t_k(P_i)\in X$ be the \emph{$k$-th ranked} (from top) alternative according to $P_i$; that is, for each $k=\{1,\ldots,m\}$, $\big|\{x\in X \mid x\mathrel{R_i} t_k(P_i)\}\big|=k$.
Given any $f : \mathcal{P}^n \to X$ and $P_i, P'_i \in \mathcal{P}$, we say that $\tilde P_{-i} \in \mathcal{P}^{n-1}$ is \emph{worst} (resp. \emph{best}) \emph{for} $(P_i,f(P'_i,\cdot))$ if for each $P_{-i} \in \mathcal{P}^{n-1}$,
\[f(P'_i,P_{-i}) \mathrel{R_i} f(P'_i,\tilde P_{-i})\;\big(\text{resp. }f(P'_i,\tilde P_{-i}) \mathrel{R_i} f(P'_i,{P}_{-i})\big).\]

\subsection{Worst-case strategy-proofness}
We first define two existing non-manipulability axioms, and then introduce a new one: worst-case strategy-proofness.

\begin{dfn}
A voting rule $ f : \mathcal{P}^n \to X $ is \textbf{strategy-proof} if for each $i\in N$, each $ P_i , P'_i \in \mathcal{P} $, and each $P_{-i}\in \mathcal{P}^{n-1}$,
\[f(P_i, P_{-i}) \mathrel{R_i} f(P'_i, P_{-i}).\]
\end{dfn}

Strategy-proofness guarantees that truthful reporting is always optimal, but it is stringent. Hence, we need a weaker axiom than strategy-proofness.

\textcite{troyan2020} introduce the following relaxation of strategy-proofness:\footnote{\textcite{troyan2020} restricts the imposition of NOM-best and NOM-worst to profitable $P'_i$. However, as any non-profitable $P'_i$ necessarily satisfies both conditions, this is equivalent to imposing them for all $P'_i$. }

\begin{dfn}[\cite{troyan2020}]
A voting rule $ f : \mathcal{P}^n \to X $ is \textbf{not obviously manipulable (NOM)} if the following two conditions hold:
\begin{enumerate}
\item \textbf{NOM-worst } For each $i\in N$ and each $P_i, P'_i \in \mathcal{P} $, whenever $\tilde P_{-i}\in \mathcal{P}^{n-1}$ is worst for $(P_i,f(P_i,\cdot))$ and $\tilde P'_{-i}\in \mathcal{P}^{n-1}$ is worst for $(P_i,f(P'_i,\cdot))$,
\[f(P_i,\tilde P_{-i}) \mathrel{R_i} f(P'_i,\tilde P'_{-i}).\]
\item \textbf{NOM-best } For each $i\in N$ and each $ P_i,P'_i \in \mathcal{P} $, whenever $\tilde P_{-i}\in \mathcal{P}^{n-1}$ is best for $(P_i,f(P_i,\cdot))$ and $\tilde P'_{-i}\in \mathcal{P}^{n-1}$ is best for $(P_i,f(P'_i,\cdot))$,
\[f(P_i, \tilde P_{-i}) \mathrel{R_i} f(P'_i,\tilde P'_{-i}).\]
\end{enumerate}
\end{dfn}

NOM focuses on manipulations by agents who are uninformed about the other agents' preferences. Each agent can only compare the worst-case (resp. best-case) outcome of truthful reporting with that of misreporting and manipulates if the latter is better. NOM prohibits such \emph{obvious manipulations}.

\textcite{aziz2021obvious} show that if $n\ge3$, most voting rules, including the plurality rule, the Borda rule, and the Dowdall rule, satisfy NOM, suggesting that NOM is too weak to capture the idea of non-manipulability in a voting model. Moreover, they show that NOM-best is much weaker than NOM-worst, since every scoring rule satisfying NOM-worst also satisfies NOM-best. Their results motivate us to introduce a stronger axiom than NOM-worst.

Our new non-manipulability axiom is as follows: 

\begin{dfn}
A voting rule $ f : \mathcal{P}^n \to X $ is \textbf{worst-case strategy-proof (WCSP)} if for each $i\in N$, each $ P_i , P'_i \in \mathcal{P} $, and each $\tilde P_{-i}\in \mathcal{P}^{n-1}$ that is worst for $(P_i,f(P_i,\cdot))$, 
\[f(P_i, \tilde P_{-i}) \mathrel{R_i} f(P'_i, \tilde P_{-i}).\]
\end{dfn}

WCSP focuses on manipulations by agents who are able to form expectations about which alternative is likely to win. If an agent were to predict a \emph{worst-case scenario}—the other agents' preferences that would lead to the worst-case outcome under her truthful reporting—she might try to manipulate. WCSP ensures that truthful reporting is optimal for each agent under all their worst-case scenarios. 

\begin{figure}[tbp]
\centering
\includegraphics[width=0.6\linewidth]{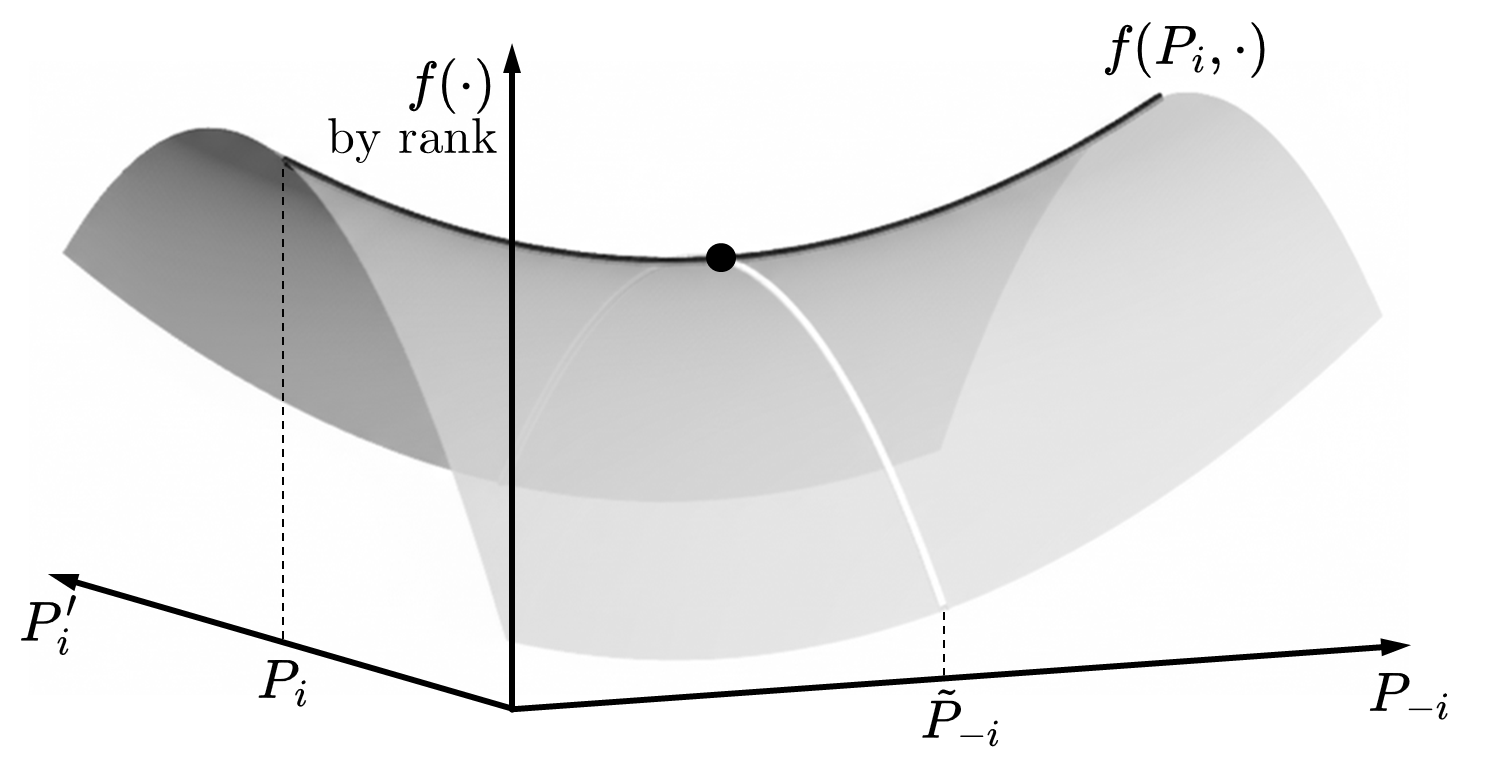}
\caption{Conceptual illustration of WCSP}
\label{fig:1}
\end{figure}

Fig. \ref{fig:1} illustrates the concept of WCSP. Given a voting rule $f:\mathcal{P}^n\to X$ and a preference $P_i\in\mathcal{P}$, the horizontal axis shows all possible preferences of the other agents $P_{-i} \in \mathcal{P}^{n-1}$; the depth axis shows all possible misreporting $P'_i \in \mathcal{P}$; and the vertical axis shows $f(\cdot)$, ordered by the rank according to $P_i$ (highest = $t_1(P_i)$). The black line represents $f(P_i, \cdot)$, with the black dot marking its lowest point at the worst-case scenario $\tilde P_{-i}\in \mathcal{P}^{n-1}$. WCSP ensures that the dot is also the highest point under the same $\tilde P_{-i}$. This point can be regarded as a saddle point in the broad sense.

Next, we clarify the relations between the three non-manipulability axioms. The following immediately follows from the definitions.

\begin{prop}\label{prop:SP_WCSP}
Strategy-proofness implies WCSP.
\end{prop}

Then, we show that WCSP is stronger than NOM-worst. The following lemma gives an equivalent formulation of NOM-worst. Here, NOM-worst is reformulated in terms of the existence (or non-existence) of other agents' preferences, rather than in terms of outcome comparisons. 

\begin{lem}\label{lem:NOMw}
The following are equivalent for any voting rule $f : \mathcal{P}^n \to X$:
\begin{enumerate}
\item $ f $ satisfies NOM-worst.
\item For each $i\in N$, each $ P_i, P'_i \in \mathcal{P} $, and each $\tilde P_{-i}\in \mathcal{P}^{n-1}$ that is worst for $(P_i,f(P_i,\cdot))$, there exists $P^\star_{-i}\in \mathcal{P}^{n-1}$ such that
\[f(P_i,\tilde P_{-i}) \mathrel{R_i} f(P'_i,P^\star_{-i}).\]
\end{enumerate}
\end{lem}

\begin{proof}
Let $ f:\mathcal{P}^n \to X$ be a voting rule.\\
$(1\Rightarrow2)$ By the definition of NOM-worst, it immediately follows. \\
$(2\Rightarrow1)$ Let $\tilde P'_{-i} \in \mathcal{P}^{n-1}$ be worst for $(P_i,f(P'_i,\cdot))$. Then, $f(P'_i,P^\star_{-i}) \mathrel{R_i} f(P'_i,\tilde P'_{-i})$. Thus, $f(P_i,\tilde P_{-i}) \mathrel{R_i} f(P'_i,\tilde P'_{-i})$.
\end{proof}

According to the definition of WCSP and Lemma \ref{lem:NOMw}, the following relation is established.

\begin{prop}\label{prop:WCSP_NOMw}
WCSP implies NOM-worst.
\end{prop}

\begin{figure}[tbp]
\centering
\includegraphics[width=0.8\linewidth]{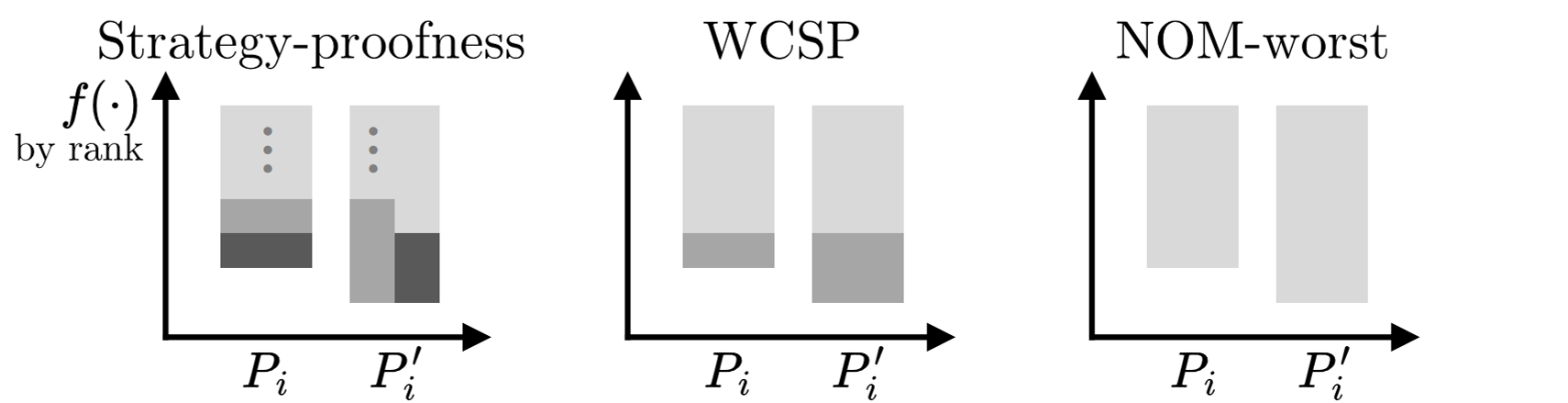}
\caption{Conceptual illustration of the three axioms}
\label{fig:2}
\end{figure}

By Propositions \ref{prop:SP_WCSP} and \ref{prop:WCSP_NOMw}, it follows that WCSP is weaker than strategy-proofness but stronger than NOM-worst.

Fig. \ref{fig:2} illustrates the corresponding non-manipulability concepts. Given a voting rule $f:\mathcal{P}^n\to X$ and a preference $P_i\in \mathcal{P}$, each square represents the range of outcomes $f(P_i,\cdot)$ or $f(P'_i,\cdot)$, where $P'_i\in \mathcal{P}$ is a misreporting. NOM-worst ensures that the worst of $f(P_i,\cdot)$ is at least as good as the worst of $f(P'_i,\cdot)$. WCSP ensures that $P_i$ is optimal in the worst-case scenario. Strategy-proofness guarantees that $P_i$ is optimal for any the other agents' preferences.

Next, we demonstrate some implications of WCSP. We show that WCSP is equivalent to strategy-proofness for \emph{tops-only} voting rules. 

\begin{dfn}
A voting rule $ f : \mathcal{P}^n \to X $ is \textbf{tops-only} if for each $i\in N$ and each $ P ,P'\in \mathcal{P}^{n}$ such that $t_1(P_i')=t_1(P_i)$, $f(P')=f(P)$.
\end{dfn}

\begin{prop}\label{prop:tops}
A tops-only voting rule is WCSP if and only if it is strategy-proof.
\end{prop}

\begin{proof}
Let $ f:\mathcal{P}^n \to X$ be a tops-only voting rule.\\
$(\Longrightarrow)$ Suppose not. Then, there exist $i\in N$, $P \in \mathcal{P}^{n}$, and $P'_i \in \mathcal{P}$ such that $f(P'_i, P_{-i}) \mathrel{P_i} f(P)$. Thus, $f(P'_i, P_{-i}) \neq f(P)$ and $f(P)\neq t_1(P_i)$. Take any $P^\star_i \in \mathcal{P}$ such that $t_1(P^\star_i)=t_1(P_i)$ and $t_m(P^\star_i)=f(P)$. Since $ f $ is tops-only, $f(P^\star_i, P_{-i})=f(P)$. Thus, $ P_{-i}$ is worst for $(P^\star_i ,f(P^\star_i,\cdot))$. Since $f(P'_i, P_{-i}) \neq f(P)$, $f(P'_i, P_{-i})\neq f(P^\star_i, P_{-i})$. Thus, $f(P'_i, P_{-i}) \mathrel{P^\star_i}f(P^\star_i, P_{-i})$. Therefore, $ f $ is not WCSP. \\
$(\Longleftarrow)$ It follows from Proposition \ref{prop:SP_WCSP}.
\end{proof}

Then, we demonstrate that WCSP and strategy-proofness are equivalent when there are only two alternatives.

\begin{prop}\label{prop:m=2}
Suppose $m = 2$. A voting rule is WCSP if and only if it is strategy-proof.
\end{prop}

\begin{proof}
Assume that $m=2$. Then, any voting rule is tops-only. By Proposition \ref{prop:tops}, any voting rule is WCSP if and only if it is strategy-proof.
\end{proof}

By Proposition \ref{prop:m=2}, when there are only two alternatives, WCSP has no new significance as an axiom. Thus, from the next section onward, we will assume the existence of at least three alternatives.

Then, we use the Gibbard–Satterthwaite theorem \citep{gibbard1973,satterthwaite1975} to examine WCSP for a \emph{onto} voting rule when there are at least three alternatives.

\begin{dfn}
A voting rule $ f : \mathcal{P}^n \to X $ is \textbf{onto}, if for each $x\in X$, there exists $ P \in \mathcal{P}^n $ such that $f(P)=x$.
\end{dfn}

\begin{dfn}
A voting rule $ f : \mathcal{P}^n \to X $ is \textbf{dictatorial}, if there exists $i \in N$ such that for each $P \in \mathcal{P}^ n$, $f (P) = t_1(P_i)$. 
\end{dfn}

\begin{lem}[\cite{gibbard1973,satterthwaite1975}]\label{lem:GS}
Suppose $m\geq3$. A voting rule is strategy-proof and onto if and only if it is dictatorial.
\end{lem}

By the definitions, the following lemma immediately follows.

\begin{lem}\label{lem:dic_tops}
Any dictatorial voting rule is tops-only.
\end{lem}

Then, we can derive the following proposition using the above lemmas.

\begin{prop}\label{prop:GS2}
Suppose $m\geq3$. A voting rule is WCSP, tops-only, and onto if and only if it is dictatorial.
\end{prop}

\begin{proof}
It follows from Lemmas \ref{lem:GS}, \ref{lem:dic_tops} and Proposition \ref{prop:tops}.
\end{proof}

Proposition \ref{prop:GS2} states an impossibility result for tops-only and onto voting rules to be WCSP.\footnote{Note that some tops-only voting rules, such as the median-voter rule \citep{moulin1980strategy}, impose restrictions on preferences and thus violate strategy-proofness and WCSP in our model.} However, since tops-only voting rules are not common, the result is not restrictive in practice. For example, the plurality rule is usually not tops-only (for details, see Section 3).

\section{Main results}
In this section, we provide the main results of this paper. We examine WCSP for standard voting rules. As discussed in Section 2, we assume that $m\ge3$ and work under this assumption throughout the section.

\begin{ass}
There exist at least three alternatives.
\end{ass} 

Now, we define a scoring rule, then introduce standard scoring rules and a scoring rule with fixed-order tie-breaking.
Let $r(P_i,x)\in\{1,\ldots, m\}$ be the \emph{rank} (from top) of $x\in X$ according to $P_i$; that is, $r(P_i,x) = \big| \{y\in X \mid y \mathrel{R_i} x \} \big|$. 
Let $\succ\:\in \mathcal{P}$ be the \emph{fixed-order} on $X$; that is, $x_1 \succ x_2 \succ \cdots \succ x_m$.
A \emph{score vector} is $s =(s(1),\ldots,s(m)) \in \mathbb{R}^m$ such that 
\[s(1)>s(m) \quad \text{and}\quad s(1) \geq s(2) \geq \cdots \geq s(m).\]
Given any scoring vector $s=(s(1),\ldots,s(m))$, let $S(P,x)$ be the \emph{total score} of $x\in X $ for $P\in\mathcal{P}^n$; that is,
\[S(P,x)= \sum_{i\in N} s(r(P_i,x)).\]

\begin{dfn}
A voting rule $ f : \mathcal{P}^n \to X $ is a \textbf{scoring rule} if there exists a score vector $s=(s(1),\ldots,s(m))$ such that for each $P\in \mathcal{P}^n$ and each $x\in X$,
\[S(P,f(P)) \geq S(P,x).\]
\end{dfn}

\begin{dfn}
A scoring rule $f : \mathcal{P}^n \to X$ associated with $s=(s(1),\ldots,s(m))$ is defined as follows:
\begin{itemize}
\item a \textbf{plurality rule} if $ s(1)> s(2) = \cdots = s(m) $,
\item an \textbf{anti-plurality rule} if $ s(1)= \cdots = s(m-1)>s(m) $,
\item a \textbf{Borda rule} if $ s(1) >s(2) $ and $s(1) -s(2)=s(2)-s(3)=\cdots=s(m-1)-s(m) $,
\item a \textbf{Dowdall rule} if there exist $\alpha\in\mathbb{R}_{++} $ and $\beta\in \mathbb{R}$ such that $s(k)=\frac{1}{k}\alpha+\beta$ for each $ k \in \{1, \ldots, m\}$.
\end{itemize}
\end{dfn}

\begin{dfn}
A scoring rule $ f : \mathcal{P}^n \to X $ associated with $s=(s(1),\ldots,s(m))$ is said to be \textbf{with fixed-order tie-breaking} (simply, \textbf{with $\bm{\succ}$}) if for each $P\in \mathcal{P}^n$ and each $x\in X\setminus \{f(P)\}$ such that $S(P,f(P)) = S(P,x)$,
\[f(P) \succ x.\]
\end{dfn}

\subsection{Anti plurality rules}
First, we examine WCSP for anti-plurality rules. Theorem \ref{thm:apv} below gives the necessary and sufficient conditions for anti-plurality rules with $\succ$ to be WCSP.

\begin{thm}\label{thm:apv}
An anti-plurality rule with $\succ$ is WCSP if and only if $n \geq 2m-1$.
\end{thm}

\begin{proof}
Let $f : \mathcal{P}^n \to X$ be an anti-plurality rule with $\succ$. Without loss of generality, assume that the score vector is given by $ s(1) = \cdots = s(m-1)=1$ and $s(m)=0$.\\
$(\Longrightarrow)$ We prove the contrapositive: if $n<2m-1$, then no anti-plurality rule with $\succ$ is WCSP. Suppose that $n<2m-1$. Let $P_i\in \mathcal{P}$ be such that $t_1(P_i)=x_2$, $t_{m-1}(P_i)=x_1$, and $t_m(P_i)=x_m$.\\\\
\textbf{Step 1. Showing $\bm{f(P)\neq t_m(P_i)}$ for each $\bm{P_{-i}}$:} \\
Take any $ P_{-i}\in\mathcal{P}^{n-1}$. Since $n<2m-1$ and $n,m$ are integers, $n+2\leq 2m $ and $(n-1)+3\le 2m.$ Then, there exist at least two distinct $a,b\in X$ such that
\begin{align*}
& \big|\{ j \in N\setminus\{i\} \mid t_m(P_j)=a \}\big|\leq 1 \quad \text{and}\\ 
& \big|\{ j \in N\setminus\{i\} \mid t_m(P_j)=b \}\big|\leq 1.
\end{align*}
Since $a,b$ are distinct, there exists $a'\in \{a,b\}$ with $a'\neq x_m$. Since $t_m(P_i)=x_m,$ 
\[\big|\{ j \in N\mid t_m(P_j)=a' \}\big|\leq 1.\]
Then,
\[S(P,a')\geq n-1.\]
Since $t_m(P_i)=x_m$,
\[S(P,x_m)\leq n-1.\]
Thus,
\[S(P,a')\geq S(P,x_m).\]
Since $a'\succ x_m$ and $f$ is with $\succ$,
\[f(P)\neq x_m.\]
Therefore, for each $P_{-i}\in\mathcal{P}^{n-1}$, $f(P)\neq t_m(P_i)$.
\\\\
\textbf{Step 2. Finding $\bm{\tilde P_{-i}}$:}\\
Let $\tilde P \in \mathcal{P}^{n}$ be such that for each $j\in N$, $t_m(\tilde P_j)=x_m$. Then, for each $x \in X \setminus \{x_m\}$,
\begin{align*}
S(P_i,\tilde P_{-i},x)&=n \quad \text{and}\\
S(P_i,\tilde P_{-i},x_m)&=0.
\end{align*}
Since $f$ is with $\succ$, by the definition of $\succ$, 
\[f(P_i,\tilde P_{-i})=x_1.\]
Since $x_1=t_{m-1}(P_i), $ $\tilde P_{-i}$ is worst for $(P_i,f(P_i,\cdot))$.\\\\
\textbf{Step 3. Completing the proof:}\\
Let $P' _{i} \in \mathcal{P}$ be such that $t_m(P'_i)=x_1$. Then, for each $x\in X\setminus\{x_1,x_m\}$,
\begin{align*}
S(P'_i,\tilde P_{-i},x)&=n,\\
S(P'_i,\tilde P_{-i},x_1)&=n-1,\quad\text{and}\\
S(P'_i,\tilde P_{-i},x_m)&=1.
\end{align*}
Then, by the definition of $\succ$,
\[f(P'_i,\tilde P_{-i})=x_2.\]
Thus,
\[f( P'_i , \tilde P_{-i})=x_2 \mathrel{P_i} x_1= f(P_i,\tilde P_{-i}).\]
Therefore, $f$ is not WCSP.\\\\
$(\Longleftarrow)$ We shall show that if $n\geq 2m-1$, then any anti-plurality rule with $\succ$ is WCSP. Assume that $n\geq 2m-1$. Take any $P_i\in \mathcal{P}$. Let $a=t_m(P_i)$.
\\\\
\textbf{Step 1. Showing the existence of $\bm{P_{-i}}$ such that $\bm{f(P)=t_m(P_i)}$:} \\
By assumption, $n-1\geq2(m-1)$.
Since $\big|N\setminus\{i\}\big|=n-1$ and $\big| X \setminus\{a\}\big|=m-1$, there exists $ P_{-i}\in\mathcal{P}^{n-1}$ such that for each $j\in N\setminus\{i\}$ and each $x\in X \setminus\{a\}$,
\[t_m(P_j)\neq a \quad \text{and} \quad \big|\{ j \mid t_m(P_j)=x \}\big|\geq 2.\]
Then, for each $x\in X \setminus\{a\}$,
\begin{align*}
&S(P, x) \leq n-2 \quad\text{and}\\
&S(P,a)=n-1. 
\end{align*}
Hence, $f(P) = a = t_m(P_i)$.\\\\
\textbf{Step 2. Completing the proof:}\\
Take any $\tilde P_{-i} \in \mathcal{P}^{n-1}$ that is worst for $(P_i,f(P_i,\cdot))$. By Step 1,
\[f(P_i, \tilde P_{-i})=a.\]
Take any $P'_i\in \mathcal{P}$. We shall prove $f(P'_i, \tilde P_{-i})=a$ by considering following two cases.
\begin{enumerate}
\item Assume that $t_m(P'_i)=a$. Then $f(P_i, \tilde P_{-i})= f(P'_i, \tilde P_{-i})=a$.
\item Assume that $t_m(P'_i)\neq a$. Take any $x\in X \setminus\{a\}$.
Then,
\begin{align*}
&S(P'_i, \tilde P_{-i}, a) > S(P_i, \tilde P_{-i}, a) \quad \text{and}\\
&S(P'_i, \tilde P_{-i}, x) \leq S(P_i, \tilde P_{-i}, x).
\end{align*}
By the definition of $f$, $S(P_i, \tilde P_{-i}, a) \geq S(P_i, \tilde P_{-i}, x)$. Then, we have $S(P'_i, \tilde P_{-i}, a) > S(P'_i, \tilde P_{-i}, x)$. Thus, $f(P'_i, \tilde P_{-i}) =a$.
\end{enumerate}
Hence, \[f(P_i, \tilde P_{-i})=a \mathrel{R_i}a=f(P'_i, \tilde P_{-i}).\]
Therefore, $f$ is WCSP.
\end{proof}

By Theorem \ref{thm:apv}, considering practical voting situations, the condition for anti-plurality rules to be WCSP is likely to be easily satisfied. For example, if there are five alternatives, nine agents should be enough.

Now, we show a result that overcomes the incompatibility between non-manipulability and ontoness, highlighted by the Gibbard–Satterthwaite theorem (Lemma \ref{lem:GS}).
First, we consider the relation between anti-plurality rules and ontoness; a condition for a reasonable voting rule.\footnote{Incidentally, there is a stronger condition than ontoness known as \emph{Pareto efficiency}, and it is well known that anti-plurality rules with $\succ$ unfortunately violate this condition.}

\begin{lem}\label{lem:apv_onto}
If $n\ge m-1$, an anti-plurality rule is onto.
Furthermore, an anti-plurality rule with $\succ$ is onto if and only if $n\ge m-1$.
\end{lem}

\begin{proof}
First, let $f:\mathcal{P}^n \to X$ be an anti-plurality rule. Suppose that $n\ge m-1$. Take any $x\in X$. Since $n\ge m-1$, we can construct $P\in\mathcal{P}^n$ such that for each $y\in X\setminus\{x\}$, $\big|\{ i\in N\mid t_m(P_i)=y\}\big|\ge 1$ and $\big|\{ i\in N\mid t_m(P_i)=x\}\big|= 0$. Then, for each $y\in X\setminus\{x\}$, $S(P,x)=n\ge n-1=S(P,y)$. Hence, $f(P)=x$. Since $x$ was arbitrary, $f$ is onto.\\\\
Next, let $f:\mathcal{P}^n\to X$ be an anti-plurality rule with $\succ$.\\
($\Longrightarrow$) Assume that $n<m-1$. We show that $f$ is not onto. Take any $P\in\mathcal{P}^n$. Since $n<m-1$, there exist at least two distinct $a,b\in X$ such that $\big|\{ i\in N\mid t_m(P_i)=a\}\big|=\big|\{ i\in N\mid t_m(P_i)=b\}\big|= 0$. Hence, for each $x\in X$,
$S(P,a)=S(P,b)\ge S(P,x)$. Without loss of generality, suppose that $a\succ b$. Then, $a\succ x_m$. Since $f$ is with $\succ$, $f(P)\ne x_m$. Therefore, $f$ is not onto.\\
($\Longleftarrow$)This follows from the first part of the proof.
\end{proof}

According to Lemma \ref{lem:apv_onto}, anti-plurality rules with $\succ$ are onto when $n \geq m-1$. Then, together with Theorem \ref{thm:apv}, the following proposition immediately follows. 

\begin{prop}\label{prop:WCSP_onto}
An anti-plurality rule with $\succ$ is WCSP and onto if and only if $n \geq 2m-1$.
\end{prop}

Proposition \ref{prop:WCSP_onto} may be viewed as partially overcoming the incompatibility highlighted by the Gibbard–Satterthwaite theorem (Lemma \ref{lem:GS}).

\subsection{Plurality rules}
Next, we examine WCSP for plurality rules. Since plurality rules with $\succ$ are tops-only, onto, and not dictatorial, the following immediately follows from Proposition \ref{prop:GS2}.

\begin{rem}
No plurality rule with $\succ$ is WCSP.
\end{rem}

However, note that plurality rules are not always tops-only. It depends on the tie-breaking procedure. We now consider every voting rule within the class of plurality rules. Theorem \ref{thm:pv} below demonstrates a impossibility result for plurality rules to be WCSP.

\begin{thm}\label{thm:pv}
For all $m\ge3$, a plurality rule is WCSP if and only if it is dictatorial and $n=2$.
\end{thm}

In Theorem \ref{thm:pv}, for clarity, we separate the cases of being dictatorial and $n=2$; however, note that plurality rules can be dictatorial only when $n=2$.

\begin{lem}
If there exists a plurality rule that is dictatorial, then $n=2$.
\end{lem}

\begin{proof}
Let $f : \mathcal{P}^n \to X$ be a dictatorial plurality rule. Without loss of generality, assume that for each $P\in\mathcal{P}^n$, $f(P)=t_1(P_i)$. Suppose not; that is, $n\ge 3$. Take any two distinct $x,y \in X.$ Let $P'\in\mathcal{P}^n$ be such that $t_1(P'_i)=x$ and $t_1(P'_j)=y$ for each $j\in N\setminus\{i\}$. Since $n\ge3$, $S(P',x)= 1$ and $S(P',y)\ge 2$. Since $f$ is a plurality rule, $f(P')=y$. Hence, $f(P')\neq x =t_1(P'_i)$ and which contradicts that $f$ is dictatorial. 
Thus, if a plurality rule is dictatorial, then $n=2$.
\end{proof}

The proof of Theorem \ref{thm:pv} is in a series of lemmas.

\begin{lem}\label{lem:pv_n2_nondic}
Suppose $n=2$. If a plurality rule $f:\mathcal{P}^2\to X$ is not dictatorial, then there exist distinct $P,P',P''\in \mathcal{P}^2$ and distinct $i,j\in N$ such that $f(P)=t_1(P_i)\neq t_1(P_j)$, $f(P')=t_1(P'_j)\neq t_1(P'_i)$, $f(P'')=t_1(P''_i)\neq t_1(P''_j)$, and $f(P'')\notin \{f(P),f(P')\}$.
\end{lem}

\begin{proof}
Assume that $n=2$. Let $f : \mathcal{P}^2\to X$ be a plurality rule and not dictatorial. Since $f$ is not dictatorial, by the contrapositive of the definition, for each $i\in N$, there exists $P^\star \in \mathcal{P}^2$ such that $f(P^\star)\neq t_1(P^\star_i)$. Since $n=2$, there exist $P,P' \in \mathcal{P}^2$ such that $f(P)\neq t_1(P_2)$ and $f(P')\neq t_1(P'_1)$. Since $f$ is a plurality rule, $f(P)=t_1(P_1)\neq t_1(P_2)$, $f(P')=t_1(P'_2)\neq t_1(P'_1)$, and $P\neq P'$. Since $m\ge3$ and $f$ is a plurality rule, there exist $P''\in \mathcal{P}$ and distinct $i,j\in N$ such that $f(P'')\notin \{f(P),f(P')\}$ and $f(P'')=t_1(P''_i)\neq t_1(P''_j)$.
\end{proof}

\begin{lem}\label{lem:pv_n2_WCSP}
Suppose $n=2$. If a plurality rule $f:\mathcal{P}^2\to X$ is WCSP, then for each distinct $i, j\in N$, each $P\in\mathcal{P}^2$ such that $f(P)=t_1(P_i)\neq t_1(P_j)$, and each $P'_j\in\mathcal{P}$, $f(P_i,P'_j)=t_1(P_i)$.
\end{lem}

\begin{proof}
Assume that $n=2$. Let $f : \mathcal{P}^2\to X$ be a plurality rule and WCSP. Take any distinct $i,j\in N$ and $P\in\mathcal{P}^2$ such that $f(P)=t_1(P_i)\neq t_1(P_j)$.
Without loss of generality, assume that $i=1$ and $j=2$. Let $a=t_1(P_1)$ and $b=t_1(P_2)$. Take any $P^\star\in\mathcal{P}^2$ such that $t_1(P^\star_1)=a$, $t_m(P^\star_1)=b$, $t_1(P^\star_2)=b$, and $t_m(P^\star_2)=a$. Since $f$ is a plurality rule, $f(P^\star)\in\{a,b\}$.
\begin{enumerate}
\item Assume that $f(P^\star)=a$.
Then, $P^\star_1$ is worst for $(P^\star_2,f(P^\star_2,\cdot))$.
Since $f$ is WCSP, for each $P'_2\in\mathcal{P}$, $f(P^\star)\mathrel{R^\star_2}f(P^\star_1,P'_2)$ and $f(P^\star_1,P'_2)=a$.
\begin{itemize}
\item Take any $P^{m}_2\in\mathcal{P}$ with $t_1(P^{m}_2)=t_{m}(P_1)$. Then, $f(P_1,P^m_2)\in\{a,t_{m}(P_1)\}$. Assume that $f(P_1,P^m_2)=t_{m}(P_1)$. Then, $P^m_2$ is worst for $(P_1,f(P_1,\cdot))$. Since $f(P^\star_1,P^m_2)=a$, $f(P^\star_1,P^m_2)\mathrel{P_1}f(P_1,P^m_2)$. Thus, $f$ is not WCSP. It is a contradiction. Hence, $f(P_1,P^m_2)=a.$ Also, for each $P'_2\in\mathcal{P}$, $f(P_1,P'_2)\ne t_m(P_1).$
\item Take any $P^{m-1}_2\in\mathcal{P}$ with $t_1(P^{m-1}_2)=t_{m-1}(P_1)$. Then, $f(P_1,P^{m-1}_2)\in\{a,t_{m-1}(P_1)\}$. Assume that $f(P_1,P^{m-1}_2)=t_{m-1}(P_1)$. By the above, for each $P'_2\in\mathcal{P}$, $f(P_1,P'_2)\ne t_m(P_1).$ Then, $P^{m-1}_2$ is worst for $(P_1,f(P_1,\cdot))$. Similarly, it leads to a contradiction. Hence, $f(P_1,P^{m-1}_2)=a.$ Also, for each $P'_2\in\mathcal{P}$, $f(P_1,P'_2)\ne t_{m-1}(P_1).$
\item Take any $P^{k}_2\in\mathcal{P}$ with $t_1(P^{k}_2)=t_{k}(P_1)$ for each $k\in\{2,\ldots ,m-2\}$. Similarly, we have $f(P_1,P^k_2)=a$ and for each $P'_2\in\mathcal{P}$, $f(P_1,P'_2)\ne t_{k}(P_1).$
\item Take any $P^{1}_2\in\mathcal{P}$ with $t_1(P^1_2)=t_{1}(P_1)=a.$ Then, $f(P_1,P^1_2)=a$.
\end{itemize}
Thus, for each $P'_2\in\mathcal{P}$, $f(P_1,P'_2)=a$.
\item Assume that $f(P^\star)=b$. Similarly, we have for each $P'_1\in\mathcal{P}$, $f(P'_1,P_2)=b$. It contradicts to $f(P)=a$.
\end{enumerate}
Therefore, for each $P'_2\in\mathcal{P}$, $f(P_1,P'_2)=a$. The same argument applies to any distinct $i,j\in N$: that is, for each $P'_j\in\mathcal{P}$, $f(P_i,P'_j)=t_1(P_i)$.
\end{proof}

\begin{lem}\label{lem:pv_n2_WCSP_dic}
Suppose $n=2$. If a plurality rule is WCSP, then it is dictatorial.
\end{lem}

\begin{proof}
Assume that $n=2$. Let $f : \mathcal{P}^2\to X$ be a plurality rule and WCSP. Suppose not. By Lemma \ref{lem:pv_n2_nondic}, there exist distinct $P,P',P''\in \mathcal{P}^2$ and distinct $i,j\in N$ such that $f(P)=t_1(P_i)\neq t_1(P_j)$, $f(P')=t_1(P'_j)\neq t_1(P'_i)$, $f(P'')=t_1(P''_i)\neq t_1(P''_j)$, and $f(P'')\notin \{f(P),f(P')\}$. Since $n=2$ and $f$ is a plurality rule and WCSP, by Lemma \ref{lem:pv_n2_WCSP}, for each $P^\star\in\mathcal{P}^2$, $f(P_i,P^\star_j)=t_1(P_i)$, $f(P^\star_i,P'_j)=t_1(P'_j)$, and $f(P''_i,P^\star_j)=t_1(P''_i)$. Thus, $f(P_i,P'_j)=t_1(P_i)=t_1(P'_j)$ and $f(P''_i,P'_j)=t_1(P''_i)=t_1(P'_j)$. Hence, $f(P)=t_1(P_i)=t_1(P''_i)=f(P'')$ and it contradicts to $ f(P'')\notin \{f(P),f(P')\}$. Therefore, $f$ is dictatorial.
\end{proof}

\begin{lem}\label{lem:pv_n3}
Suppose $n\ge3$. No plurality rule is WCSP.
\end{lem}

\begin{proof}
See Appendix \ref{app:pv_n3}.
\end{proof}

\begin{proof}[Proof of Theorem \ref{thm:pv}]
($\Longrightarrow$) It follows from Lemmas \ref{lem:pv_n2_WCSP_dic} and \ref{lem:pv_n3}.\\
($\Longleftarrow$) It follows form Proposition \ref{prop:GS2}.
\end{proof}

According to \textcite{aziz2021obvious}, any plurality rule with $n\ge3$ satisfies NOM. However, Theorem \ref{thm:pv} demonstrates that plurality rules can be considered vulnerable to manipulations.

\subsection{Other scoring rules}
Then, we examine WCSP for Borda rules and Dowdall rules. Theorems \ref{thm:borda} and \ref{thm:dowdall} below show that Borda rules and Dowdall rules are not WCSP.

\begin{thm}\label{thm:borda} 
 For all $n\ge 2$ and all $m\ge3$, no Borda rule is WCSP.
\end{thm}

\begin{proof}
See Appendix \ref{app:borda}.
\end{proof}

\begin{thm}\label{thm:dowdall}
For all $n\ge 2$ and all $m\ge3$, no Dowdall rule is WCSP.
\end{thm}

\begin{proof}
See Appendix \ref{app:dowdall}.
\end{proof}

According to \textcite{aziz2021obvious}, any Borda rule and any Dowdall rule with $n \ge 3$ are NOM. However, Theorems \ref{thm:borda} and \ref{thm:dowdall} demonstrate that these rules are not WCSP and vulnerable to manipulations.

Next, we examine WCSP for general scoring rules. Theorem \ref{thm:nonapv} below shows that among the class of scoring rules with $\succ$, no rule other than anti-plurality rules satisfies WCSP when $n \geq m$.

\begin{thm}\label{thm:nonapv}
Suppose $n\geq m$. No scoring rule with $\succ $, other than anti-plurality rules, is WCSP.
\end{thm}

\begin{proof}
See Appendix \ref{app:nonapv}.
\end{proof}

The following corollary is an immediate consequence of Theorems \ref{thm:apv} and \ref{thm:nonapv}.

\begin{cor}
Suppose $n\ge 2m-1$.
A scoring rule with $\succ $ is WCSP if and only if it is an anti-plurality rule.
\end{cor}

\begin{table}[tbp]
\centering
\begin{tabular}{|c|c|c|c|} \hline
\multirow{5}{*}{Scoring rule}&Anti-plurality rule with $\succ$&$n\geq 2m-1 \Leftrightarrow +$ &Theorem \ref{thm:apv}\\ \cline{2-4}
&Plurality rule&dictatorial $\land$ $n=2$ $\Leftrightarrow +$ &Theorem \ref{thm:pv}\\ \cline{2-4}
&Borda rule&$-$&Theorem \ref{thm:borda}\\ \cline{2-4}
&Dowdall rule&$-$&Theorem \ref{thm:dowdall}\\ \cline{2-4}
&with $\succ$ (other than anti-plurality rule)& $n\geq m \Rightarrow -$&Theorem \ref{thm:nonapv}\\ \hline
\end{tabular}
\caption{Satisfaction of WCSP with at least three alternatives ($m\geq3$)}
\label{table:m3} 
\end{table}

\begin{table}[tbp]
\centering
\begin{tabular}{|c|c|}\hline
Strategy-proofness $\Rightarrow$ WCSP $\Rightarrow$ NOM-worst &Propositions \ref{prop:SP_WCSP} and \ref{prop:WCSP_NOMw}\\ \hline
$m=2$ $\Rightarrow$ tops-only $\Rightarrow$ (strategy-proofness $\Leftrightarrow$ WCSP)&Propositions \ref{prop:tops} and \ref{prop:m=2}\\ \hline
$m\ge3$ $\Rightarrow$ (dictatorial $\Leftrightarrow$ WCSP $\land$ tops-only $\land$ onto)&Proposition \ref{prop:GS2}\\ \hline
Anti-plurality rule with $\succ$ $\Rightarrow$ ($n\ge2m-1$ $\Leftrightarrow$ WCSP $\land$ onto)&Proposition \ref{prop:WCSP_onto}\\ \hline
\end{tabular}
\caption{Implication of WCSP}
\label{table:im}
\end{table}

\section{Conclusion}
We introduced a new non-manipulability axiom called \emph{worst-case strategy-proofness (WCSP)}. This axiom was motivated by the need to advocate for a less stringent condition to \emph{strategy-proofness}, which is somewhat restrictive in voting, while strengthening the robustness axiom against \emph{obvious manipulations} \citep{troyan2020}. We examined the implications of WCSP in a voting model. Tables \ref{table:m3} and \ref{table:im} summarize the results. We found a necessary and sufficient condition for the anti-plurality rule with fixed-order tie-breaking to satisfy WCSP. We have also observed that standard voting rules including the plurality rule, the Borda rule, and the Dowdall rule violate WCSP. Future research could examine WCSP for other voting rules, or investigate WCSP in other models.

\newpage
\printbibliography

\appendix
\section{Appendix}
\subsection{Proof of Lemma \ref{lem:pv_n3}}\label{app:pv_n3}
We shall show that if $n\ge3$, then no plurality rule is WCSP. Assume that $n\geq3$. Let $f : \mathcal{P}^n\to X$ be a plurality rule. Take any distinct $a,b,c\in X$. There are three cases to consider.\\\\
\textbf{Case 1. $\bm{n=3k+3}$ with $\bm {k\in \mathbb{Z}_{\geq 0}}$:}\\
Let $P \in \mathcal{P}^{3k+3}$ be given by the following table:\footnote{In the table, for each agent $i \in N$, $t_1(P_i), t_2(P_i), \ldots, t_m(P_i)$ are arranged from top to bottom. The same background color represents the same order across the agents, while light gray and dark gray indicate opposite orders.\label{fot:table}}
\begin{table}[H]
\centering
\begin{tabular}{ccc|ccc|ccc|ccc}
$P_1$ &$P_2$&$P_3$&$P_4$&$\cdots$&$P_{k+3}$& $P_{k+4}$& $\cdots$& $P_{2k+3}$& $P_{2k+4}$& $\cdots$&$P_{3k+3}$\\ \hline
$a$& $b$&$c$ & $a$& $\cdots$& $a$& $b$& $\cdots$& $b$& $c$ & $\cdots$&$c$ \\
$b$&$c$&$a$ & $\vdots$& $\cdots$& $\vdots$& $\vdots$& $\cdots$& $\vdots$& $\vdots$& $\cdots$&$\vdots$\\
$\vdots$& $\vdots$& $\vdots$& $\vdots$& $\cdots$& $\vdots$& $\vdots$& $\cdots$& $\vdots$& $\vdots$& $\cdots$&$\vdots$\\
$c$& $a$&$b$ & $\vdots$& $\cdots$& $\vdots$& $\vdots$& $\cdots$& $\vdots$& $\vdots$& $\cdots$&$\vdots$\\
\multicolumn{3}{c}{}& \multicolumn{3}{c}{\upbracefill}& \multicolumn{3}{c}{\upbracefill}& \multicolumn{3}{c}{\upbracefill}\\
\multicolumn{3}{c}{}& \multicolumn{3}{c}{$k$ agents}& \multicolumn{3}{c}{$k$ agents}& \multicolumn{3}{c}{$k$ agents}\\
\end{tabular}
\end{table} 
Then, for each $x\in X \setminus \{a,b,c\}$, $S(P,a)=S(P,b)=S(P,c)>S(P,x)$ and $f(P)\in\{a,b,c\}$. Let $P'\in\mathcal{P}^n$ be such that $t_1(P'_1)=b$, $t_1(P'_2)=c$, and $t_1(P'_3)=a$.
\begin{enumerate}
\item If $f(P)=a$, then $P_{-2}$ is worst for $(P_2,f(P_2,\cdot))$ and $f(P'_2,P_{-2})=c\mathrel{P_2}a=f(P)$.
\item If $f(P)=b$, then $P_{-3}$ is worst for $(P_3,f(P_3,\cdot))$ and $f(P'_3,P_{-3})=a\mathrel{P_3}b=f(P)$.
\item If $f(P)=c$, then $P_{-1}$ is worst for $(P_1,f(P_1,\cdot))$ and $f(P'_1,P_{-1})=b\mathrel{P_1}c=f(P)$.
\end{enumerate}
Therefore, $f$ is not WCSP.\\\\
\textbf{Case 2. $\bm{n=3k+4}$ with $\bm {k\in \mathbb{Z}_{\geq 0}}$:}\\
Let $P \in \mathcal{P}^{3k+4}$ and $P'_1,P'_3\in \mathcal{P}$ be given by the following table:\footref{fot:table}
\begin{table}[H]
\centering
\begin{tabular}{cccc|ccc|ccc|ccc||cc}
$P_1$ &$P_2$&$P_3$ &$P_4$ &$P_5$&$\cdots$&$P_{k+4}$& $P_{k+5}$& $\cdots$& $P_{2k+4}$& $P_{2k+5}$& $\cdots$&$P_{3k+4}$&$P'_1$&$P'_3$\\ \hline
$b$&$b$& $c$&$c$ & $a$& $\cdots$& $a$& $b$& $\cdots$& $b$& $c$ & $\cdots$&$c$&$a$&$a$\\
$\vdots$&$\vdots$&$\vdots$&$\vdots$& $\vdots$& $\cdots$& $\vdots$& $\vdots$& $\cdots$& $\vdots$& $\vdots$& $\cdots$&$\vdots$&$b$&$c$\\
$\vdots$& $\vdots$& $\vdots$& $\vdots$& $\vdots$& $\cdots$& $\vdots$& $\vdots$& $\cdots$& $\vdots$& $\vdots$& $\cdots$& $\vdots$& $\vdots$&$\vdots$\\
$\vdots$& $\vdots$& $\vdots$& $\vdots$& $\vdots$& $\cdots$& $\vdots$& $\vdots$& $\cdots$& $\vdots$& $\vdots$& $\cdots$& $\vdots$& $c$&$b$\\
\multicolumn{4}{c}{}& \multicolumn{3}{c}{\upbracefill}& \multicolumn{3}{c}{\upbracefill}& \multicolumn{3}{c}{\upbracefill}&\multicolumn{2}{c}{}\\
\multicolumn{4}{c}{}& \multicolumn{3}{c}{$k$ agents}& \multicolumn{3}{c}{$k$ agents}& \multicolumn{3}{c}{$k$ agents}&\multicolumn{2}{c}{}\\
\end{tabular}
\end{table} 
Then, for each $x\in X \setminus \{b,c\}$, $S(P,b)=S(P,c)>S(P,x)$ and $f(P)\in\{b,c\}$.
\begin{enumerate}
\item Assume that $f(P)=b$. Then, for each $x\in X \setminus \{c\}$, $S(P'_1,P_{-1},c)>S(P'_1,P_{-1},x)$ and $f(P'_1,P_{-1})=c$. Thus, $P_{-1}$ is worst for $(P'_1,f(P'_1,\cdot))$ and $f(P)=b\mathrel{P'_1}c=f(P'_1,P_{-1})$.
\item Assume that $f(P)=c$. Then, for each $x\in X \setminus\{b\}$, $S(P'_3,P_{-3},b)>S(P'_3,P_{-3},x)$ and $f(P'_3,P_{-3})=b$. Thus, $P_{-3}$ is worst for $(P'_3,f(P'_3,\cdot))$ and $f(P)=c\mathrel{P'_3}b=f(P'_3,P_{-3})$.
\end{enumerate}
Then, $f$ is not WCSP.\\\\
\textbf{Case 3. $\bm{n=3k+5}$ with $\bm {k\in \mathbb{Z}_{\geq 0}}$:}\\
Let $P \in \mathcal{P}^{3k+5}$ and $P'_1,P'_2,P'_5\in \mathcal{P}$ be given by the following table:\footref{fot:table}
\begin{table}[H]
\centering
\begin{tabular}{ccccc|ccc|ccc|ccc||ccc}
$P_1$ &$P_2$ & $P_3$&$P_4$&$P_5$&$P_6$&$\cdots$&$P_{k+5}$& $P_{k+6}$& $\cdots$& $P_{2k+5}$& $P_{2k+6}$& $\cdots$&$P_{3k+5}$ & $P'_1$& $P'_2$ &$P'_5$\\ \hline
$a$& $a$& $b$&$b$&$c$ & $a$& $\cdots$& $a$& $b$& $\cdots$& $b$& $c$ & $\cdots$&$c$ & $b$& $c$&$a$\\
$b$&$\vdots$& $\vdots$&$\vdots$ &$a$ & $\vdots$& $\cdots$& $\vdots$& $\vdots$& $\cdots$& $\vdots$& $\vdots$& $\cdots$&$\vdots$ & $\vdots$& $a$ &$\vdots$\\
$\vdots$& $\vdots$& $\vdots$&$\vdots$ & $\vdots$& $\vdots$& $\cdots$& $\vdots$& $\vdots$& $\cdots$& $\vdots$& $\vdots$& $\cdots$&$\vdots$ & $\vdots$& $\vdots$ &$\vdots$\\
$c$& $\vdots$& $\vdots$&$\vdots$&$b$ & $\vdots$& $\cdots$& $\vdots$& $\vdots$& $\cdots$& $\vdots$& $\vdots$& $\cdots$&$\vdots$ & $\vdots$& $b$ &$\vdots$\\
\multicolumn{5}{c}{}& \multicolumn{3}{c}{\upbracefill}& \multicolumn{3}{c}{\upbracefill}& \multicolumn{3}{c}{\upbracefill} & \multicolumn{3}{c}{}\\
\multicolumn{5}{c}{}& \multicolumn{3}{c}{$k$ agents}& \multicolumn{3}{c}{$k$ agents}& \multicolumn{3}{c}{$k$ agents} & \multicolumn{3}{c}{}\\
\end{tabular}
\end{table} 
Let $P_{-(1,2)}=(P_3,\ldots,P_{3k+5})$. For each $x\in X \setminus \{a,b,c\}$, $S(P,a)=S(P,b)>S(P,c)>S(P,x)$ and $S(P'_2,P_{-2},b)=S(P'_2,P_{-2},c)>S(P'_2,P_{-2},a)>S(P'_2,P_{-2},x)$. Hence, $f(P)\in\{a,b\}$ and $f(P'_2,P_{-2})\in\{b,c\}$.
\begin{enumerate}
\item Assume that $f(P)=b$. Then, $P_{-5}$ is worst for $(P_5,f(P_5,\cdot))$. For each $x\in X \setminus \{a\}$, $S(P'_5,P_{-5},a)>S(P'_5,P_{-5},x)$ and $f(P'_5,P_{-5})=a\mathrel {P_5}b=f(P)$.
\item Assume that $f(P)=a$ and $f(P'_2,P_{-2})=b$. Then, $P_{-2}$ is worst for $(P'_2,f(P'_2,\cdot))$. By assumption, $f(P)=a \mathrel {P'_2}b= f(P'_2,P_{-2})$.
\item Assume that $f(P)=a$ and $f(P'_2,P_{-2})=c$. Then, $(P'_{2},P_{-(1,2)})$ is worst for $(P_1,f(P_1,\cdot))$. For each $x\in X \setminus \{b\}$, $S(P'_1,P'_{2},P_{-(1,2)},b)>S(P'_1,P'_{2},P_{-(1,2)},x)$ and $f(P'_1,P'_{2},P_{-(1,2)})=b \mathrel {P_1}c= f(P'_2,P_{-2})$.
\end{enumerate}
Therefore, $f$ is not WCSP.\\\\
Considering the three cases, we conclude that if $n\ge3$, then no plurality rule $f$ is WCSP.
\hfill \qed

\subsection{Proof of Theorem \ref{thm:borda}\label{app:borda}}
We shall show that no Borda rule is WCSP. Let $f : \mathcal{P}^n \to X$ be a Borda rule. Take any distinct $a,b,c\in X$. There are six cases to consider.\\\\
\textbf{Case 1. $\bm{m=3}$ and $\bm{n=2}$:}\\
Let $P,P^\star,P',P''\in \mathcal{P}^2$ be given by the following table:\footref{fot:table}
\begin{table}[H]
\centering
\begin{tabular}{cc||cc||cc||cc}
$P_1$&$P_2$&$P^\star_1$&$P^\star_2$&$P'_1$&$P'_2$&$P''_1$&$P''_2$\\\hline
$a$&$c$&$c$&$b$&$b$&$b$&$a$&$a$\\
$b$&$b$&$a$&$a$&$a$&$c$&$c$&$b$\\
$c$&$a$&$b$&$c$&$c$&$a$&$b$&$c$\\
\end{tabular}
\end{table}
Then, $S(P,a)=S(P,b)=S(P,c)$, $S(P^\star,a)=S(P^\star,b)=S(P^\star,c)$, and $S(P_1,P^\star_2,a)=S(P_1,P^\star_2,b)>S(P_1,P^\star_2,c)$. Hence, $f(P)\in\{a,b,c\}$, $f(P^\star)\in\{a,b,c\}$, and $f(P_1,P^\star_2)\in\{a,b\}$. Since $f$ is a Borda rule, for each $\tilde P_2\in\mathcal{P}\setminus\{P_2\}$ and each $\tilde P^\star_1\in\mathcal{P}\setminus\{P^\star_1\}$, $f(P_1,\tilde P_2)\neq c$ and $f(\tilde P^\star_1,P^\star_2)\neq c$.
\begin{enumerate}
\item If $f(P)=c$, then $P_{2}$ is worst for $(P_1,f(P_1,\cdot))$. Then, $S(P'_1,P_2,b)>S(P'_1,P_2,c)>S(P'_1,P_2,a)$ and $f(P'_1,P_2)=b\mathrel {P_1}c=f(P)$.
\item If $f(P)=a$, then $P_1$ is worst for $(P_2,f(P_2,\cdot))$ and $f(P_1,P'_2)=b\mathrel {P_2}a=f(P)$.
\item If $f(P^\star)=b$, then, $P^\star_2$ is worst for $(P^\star_1,f(P^\star_1,\cdot))$ and $f(P''_1,P^\star_2)=a\mathrel {P^\star_1}b=f(P^\star)$.
\item If $f(P^\star)=c$, then, $P^\star_1$ is worst for $(P^\star_2,f(P^\star_2,\cdot))$ and $f(P^\star_1,P''_2)=a \mathrel {P^\star_2}c=f(P^\star)$.
\item If $f(P)=b$, $f(P^\star)=a$ and $f(P_1,P^\star_2)=b$, then $P^\star_2$ is worst for $(P_1,f(P_1,\cdot))$. Also, $f(P''_1,P^\star_2)=a \mathrel {P_1}b= f(P_1,P^\star_2)$.
\item If $f(P)=b$, $f(P^\star)=a$ and $f(P_1,P^\star_2)=a$, then $P_1$ is worst for $(P^\star_2,f(P^\star_2,\cdot))$. Also, $f(P_1,P'_2)=b \mathrel {P^\star_2}a= f(P_1,P^\star_2)$.
\end{enumerate}
Therefore, $f$ is not WCSP.\\\\
\textbf{Case 2. $\bm{m\geq4}$ and $\bm{n=2}$:}\\
Take any $d \in X$ with $d\neq a,b,c$. Let $P,P'\in \mathcal{P}^2$ and $P'_1\in \mathcal{P}$ be given by the following table:\footref{fot:table}
\begin{table}[H]
\centering
\begin{tabular}{cc||cc||c}
$P_1$ &$P_2$ &$P'_1$ &$P'_2$&$P''_1$ \\ \hline
$a$&$d$ &$b$ &$c$&$b$\\
$b$&$c$ &$a$ &$d$&$a$\\
\cellcolor{lightgray}$\vdots$& \cellcolor{gray}$\vdots$ & \cellcolor{lightgray}$\vdots$ &\cellcolor{gray}$\vdots$& \cellcolor{lightgray}$\vdots$\\
$c$&$b$ &$c$ &$b$&$d$\\
$d$& $a$ & $d$ &$a$& $c$ \\
\end{tabular}
\end{table}
Then, for each $x,y\in X$, $S(P,x)=S(P,y)$. Hence, $f(P)\in X$. Since $f$ is a Borda rule, for each $\tilde P_2\in\mathcal{P}\setminus\{P_2\}$, $f(P_1,\tilde P_2)\neq d$.
\begin{enumerate}
\item Assume that $f(P)=d$. Then, $P_2$ is worst for $(P_1,f(P_1,\cdot))$ and $f(P'_1,P_2)=b\mathrel {P_1}d=f(P)$.
\item Assume that $f(P)\neq d$. Since $f(P_1,P'_2)=c=t_{m-1}(P_1)$, $P'_2$ is worst for $(P_1,f(P_1,\cdot))$. Then, $f(P''_1,P'_2)=b \mathrel {P_1}c=f(P_1,P'_2)$.
\end{enumerate}
Therefore, $f$ is not WCSP.\\\\
\textbf{Case 3. $\bm{m=3}$ and $\bm{n=2k+3}$ with $\bm {k\in \mathbb{Z}_{\geq 0}}$:}\\
Let $P \in \mathcal{P}^{2k+3}$ and $P'_1,P'_2,P'_3\in \mathcal{P}$ be given by the following table:\footref{fot:table}
\begin{table}[H]
\centering
\begin{tabular}{ccc|ccc|ccc||ccc}
$P_1$ &$P_2$&$P_3$&$P_4$&$\cdots$&$P_{k+3}$& $P_{k+4}$& $\cdots$&$P_{2k+3}$ & $P'_1$ & $P'_2$&$P'_3$\\ \hline
$a$& $b$&$c$&$c$&$\cdots$&$c$ & $a$& $\cdots$&$a$ & $b$& $c$&$a$\\
$b$&$c$&$a$&$b$&$\cdots$&$b$& $b$& $\cdots$&$b$ & $a$& $b$&$c$\\
$c$& $a$&$b$& $a$& $\cdots$& $a$ & $c$ & $\cdots$&$c$ & $c$& $a$&$b$\\
\multicolumn{3}{c}{}& \multicolumn{3}{c}{\upbracefill} & \multicolumn{3}{c}{\upbracefill} & \multicolumn{3}{c}{}\\
\multicolumn{3}{c}{}& \multicolumn{3}{c}{$k$ agents} & \multicolumn{3}{c}{$k$ agents} & \multicolumn{3}{c}{}\\
\end{tabular}
\end{table}
Then, $S(P,a)=S(P,b)=S(P,c)$. Hence, $f(P)\in\{a,b,c\}$.
\begin{enumerate}
\item If $f(P)=a$, then $P_{-2}$ is worst for $(P_2,f(P_2,\cdot))$ and $f(P'_2,P_{-2})=c \mathrel{P_2} a=f(P)$.
\item If $f(P)=b$, then $P_{-3}$ is worst for $(P_3,f(P_3,\cdot))$ and $f(P'_3,P_{-3})=a \mathrel{P_3} b=f(P)$.
\item If $f(P)=c$, then $P_{-1}$ is worst for $(P_1,f(P_1,\cdot))$ and $f(P'_1,P_{-1})=b \mathrel{P_1} c=f(P)$.
\end{enumerate}
Therefore, $f$ is not WCSP.\\\\
\textbf{Case 4. $\bm{m=3}$ and $\bm{n=2k+4}$ with $\bm {k\in \mathbb{Z}_{\geq 0}}$:}\\
Let $P \in \mathcal{P}^{2k+4}$ and $P'_1,P'_2,P'_3\in \mathcal{P}$ be given by the following table:\footref{fot:table}
\begin{table}[H]
\centering
\begin{tabular}{cccc|ccc|ccc||ccc}
$P_1$ &$P_2$ &$P_3$&$P_4$&$P_5$&$\cdots$&$P_{k+4}$& $P_{k+5}$& $\cdots$&$P_{2k+4}$ & $P'_1$ & $P'_2$ &$P'_3$\\ \hline
$a$& $a$&$b$&$c$&$c$&$\cdots$&$c$ & $a$& $\cdots$&$a$ & $b$& $c$&$c$\\
$b$&$c$&$c$&$b$&$b$&$\cdots$&$b$& $b$& $\cdots$&$b$ & $a$& $a$ &$b$\\
$c$& $b$&$a$&$a$& $a$& $\cdots$& $a$ & $c$ & $\cdots$&$c$ & $c$& $b$ &$a$\\
\multicolumn{4}{c}{} & \multicolumn{3}{c}{\upbracefill} & \multicolumn{3}{c}{\upbracefill} & \multicolumn{3}{c}{}\\
\multicolumn{4}{c}{} & \multicolumn{3}{c}{$k$ agents} & \multicolumn{3}{c}{$k$ agents} & \multicolumn{3}{c}{}\\
\end{tabular}
\end{table}
Then, $S(P,a)=S(P,b)=S(P,c)$. Hence, $f(P)\in\{a,b,c\}$.
\begin{enumerate}
\item If $f(P)=a$, then $P_{-3}$ is worst for $(P_3,f(P_3,\cdot))$ and $f(P'_3,P_{-3})=c \mathrel{P_3} a=f(P)$.
\item If $f(P)=b$, then $P_{-2}$ is worst for $(P_2,f(P_2,\cdot))$ and $f(P'_2,P_{-2})=c \mathrel{P_2} b=f(P)$.
\item If $f(P)=c$, then $P_{-1}$ is worst for $(P_1,f(P_1,\cdot))$ and $f(P'_1,P_{-1})=b \mathrel{P_1} c=f(P)$.
\end{enumerate}
Therefore, $f$ is not WCSP.\\\\
\textbf{Case 5. $\bm{m\geq4}$ and $\bm{n=2k+3}$ with $\bm {k\in \mathbb{Z}_{\geq 0}}$:}\\
Let $P \in \mathcal{P}^{2k+3}$ be given by the following table:\footref{fot:table}
\begin{table}[H]
\centering
\begin{tabular}{ccc|ccc}
$P_1$ &$\cdots$&$P_{k+1}$&$P_{k+2}$&$\cdots$&$P_{2k+3}$\\ \hline
$a$& $\cdots$&$a$&$c$&$\cdots$&$c$\\
\cellcolor{lightgray}$\vdots$&$\cdots$&\cellcolor{lightgray}$\vdots$&$b$&$\cdots$&$b$\\
$b$&$\cdots$&$b$&\cellcolor{gray}$\vdots$&$\cdots$&\cellcolor{gray}$\vdots$\\
$c$& $\cdots$&$c$& $a$& $\cdots$& $a$\\
\multicolumn{3}{c}{\upbracefill}& \multicolumn{3}{c}{\upbracefill}\\
\multicolumn{3}{c}{$k+1$ agents}& \multicolumn{3}{c}{$k+2$ agents}\\
\end{tabular}
\end{table}
Then, for each $x\in X\setminus\{c\}$, $S(P,c)>S(P,x)$ and $S(P,c)=S(P,b)+s(1)-s(2)$. Thus, $f({P}) = c$ and $P_{-1}$ is worst for $(P_1, f(P_1,\cdot))$. 
Let $P' _{1} \in \mathcal{P}$ be such that $t_1(P'_1)=b$ and $t_m(P'_1)=c$. Then, since $s(1)-s(m-1)>s(1)-s(2)$, $S(P'_1, P_{-1},b)=S(P,b)+s(1)-s(m-1)>S(P,b)+s(1)-s(2)=S(P,c)=S(P'_1, P_{-1},c)$. Then, $f(P'_1, P_{-1})\neq c$. Hence,
\[f(P'_1, P_{-1})\mathrel {P_1} f(P)=c.\]
Therefore, $f$ is not WCSP.\\\\
\textbf{Case 6. $\bm{m\geq4}$ and $\bm{n=2k+4}$ with $\bm {k\in \mathbb{Z}_{\geq 0}}$:}\\
Take any $d \in X$ with $d\neq a,b,c$. Let $P \in \mathcal{P}^{2k+4}$ be given by the following table:\footref{fot:table}
\begin{table}[H]
\centering
\begin{tabular}{ccc|ccc|c|c}
$P_1$ &$\cdots$&$P_{k+1}$ &$P_{k+2}$&$\cdots$&$P_{2k+2}$& $P_{2k+3}$& $P_{2k+4}$\\ \hline
$a$& $\cdots$&$a$ &$d$&$\cdots$&$d$& $a$& $d$\\
$b$&$\cdots$&$b$ &$c$&$\cdots$&$c$ & $c$ & $b$\\
\cellcolor{lightgray}$\vdots$&$\cdots$&\cellcolor{lightgray}$\vdots$& \cellcolor{gray}$\vdots$&$\cdots$&\cellcolor{gray}$\vdots$&\cellcolor{lightgray}$\vdots$& \cellcolor{gray}$\vdots$\\
$c$& $\cdots$& $c$ & $b$& $\cdots$& $b$ & $d$&$c$ \\
$d$& $\cdots$&$d$&$a$& $\cdots$& $a$ & $b$& $a$\\
\multicolumn{3}{c}{\upbracefill} & \multicolumn{3}{c}{\upbracefill} & \multicolumn{2}{c}{}\\
\multicolumn{3}{c}{$k+1$ agents} & \multicolumn{3}{c}{$k+1$ agents} & \multicolumn{2}{c}{}\\
\end{tabular}
\end{table}
Then, for each $x\in X \setminus \{d\}$, $S(P,d)>S(P,x)$ and $S(P,d)=S(P,c)+s(1)-s(2)$. Thus, $f({P}) = d$ and $P_{-1}$ is worst for $(P_1, f(P_1,\cdot))$. Let $P' _{1} \in \mathcal{P}$ be such that $t_1(P'_1)=c$ and $t_m(P'_1)=d$. Then, since $s(1)-s(m-1)>s(1)-s(2)$, $S(P'_1, P_{-1},c)=S(P,c)+s(1)-s(m-1)>S(P,c)+s(1)-s(2)=S(P,d)=S(P'_1, P_{-1},d)$. Then, $f(P'_1, P_{-1})\neq d$. Hence,
\[f(P'_1, P_{-1})\mathrel {P_1} f(P)=d.\]
Therefore, $f$ is not WCSP.\\\\
Considering all six cases, we have:
\begin{itemize}
\item If $n = 2$, then $f$ is not WCSP (Cases 1 and 2).
\item If $n \geq 3$ and $m = 3$, then $f$ is not WCSP (Cases 3 and 4).
\item If $n \geq 3$ and $m \geq 4$, then $f$ is not WCSP (Cases 5 and 6).
\end{itemize}
Thus, we conclude no Borda rule $f$ is WCSP.
\hfill \qed

\subsection{Proof of Theorem \ref{thm:dowdall}\label{app:dowdall}}
We shall show that no Dowdall rule is WCSP. Let $f : \mathcal{P}^n \to X$ be a Dowdall rule. Without loss of generality, assume that the score vector is given by $s(k)=\frac{1}{k}$ for each $k\in\{1,\ldots,m\}$. Take any distinct $a,b,c\in X$. There are five cases to consider.\\\\
\textbf{Case 1. $\bm{n=2}$:} \\
Let $P \in \mathcal{P}^2$ be given by the following table:\footref{fot:table}
\begin{table}[H]
\centering
\begin{tabular}{cc}
$P_1$&$P_2$\\\hline
$a$&$c$\\
$b$&$b$\\
$\vdots$&$\vdots$\\
$c$&$a$\\
\end{tabular}
\end{table}
Then, for each $x\in X\setminus\{a,c,b\}$, $S(P,a)=S(P,c)=1+\frac{1}{m}>1=S(P,b)> S(P,x)$. Thus, $f(P)\in\{a,c\}$. Let $P'\in \mathcal{P}^2$ be such that $t_1(P'_1)=b$, $t_m(P'_1)=c$, $t_1(P'_2)=b$, and $t_m(P'_2)=a$.
\begin{enumerate}
\item If $f(P)=c$, then $P_2$ is worst for $(P_1,f(P_1,\cdot))$ and $S(P'_1,P_2,b)=1+\frac{1}{2}>1+\frac{1}{m}=S(P'_1,P_2,c)$. Then, $f(P'_1,P_2)\neq c$ and $f(P'_1,P_2)\mathrel{P_1}f(P)=c$.
\item If $f(P)=a$, then $P_1$ is worst for $(P_2,f(P_2,\cdot))$ and $S(P_1,P'_2,b)=1+\frac{1}{2}>1+\frac{1}{m}=S(P_1,P'_2,a)$. Then, $f(P_1,P'_2)\neq a$ and $f(P_1,P'_2)\mathrel{P_2}f(P)=a$.
\end{enumerate}
Therefore, $f$ is not WCSP.\\\\
\textbf{Case 2. $\bm{m=3}$ and $\bm{n=3}$:} \\
Let $P,P' \in \mathcal{P}^3$ be given by the following table:\footref{fot:table}
\begin{table}[H]
\centering
\begin{tabular}{ccc||ccc}
$P_1$&$P_2$&$P_3$&$P'_1$&$P'_2$&$P'_3$\\\hline
$a$&$c$&$b$&$b$&$a$&$c$\\
$b$&$a$&$c$&$a$&$c$&$b$\\
$c$&$b$&$a$&$c$&$b$&$a$\\
\end{tabular}
\end{table}
Then, $S(P,a)=S(P,a)=S(P,c)$ and $f(P)\in\{a,b,c\}$.
\begin{enumerate}
\item If $f(P)=c$, then $P_{-1}$ is worst for $(P_1,f(P_1,\cdot))$ and $f(P'_1,P_{-1})=b\mathrel{P_1}c=f(P)$. 
\item If $f(P)=b$, then $P_{-2}$ is worst for $(P_2,f(P_2,\cdot))$ and $f(P'_2,P_{-2})=a\mathrel{P_2}b=f(P)$. 
\item If $f(P)=a$, then $P_{-3}$ is worst for $(P_3,f(P_3,\cdot))$ and $f(P'_3,P_{-3})=c\mathrel{P_3}a=f(P)$.
\end{enumerate}
Therefore, $f$ is not WCSP\\\\
\textbf{Case 3. $\bm{m\geq4}$ and $\bm{n=3}$:}\\
Take any $d \in X$ with $d\neq a,b,c$. Let $P \in \mathcal{P}^3$ be given by the following table:\footref{fot:table}
\begin{table}[H]
\centering
\begin{tabular}{ccc}
$P_1$&$P_2$&$P_3$\\\hline
$a$&$d$&$c$\\
$b$&$b$&$d$\\
$c$& $a$& $b$\\
$\vdots$&$\vdots$&$\vdots$\\
$d$&$c$&$a$\\
\end{tabular}
\end{table}
Then, for each $x\in X\setminus \{a,b,c,d\}$,
\begin{align*}
S(P,d)= s(m)+s(1)+s(2)&=1+\frac{1}{2}+\frac{1}{m},\\
S(P,a)=S(P,c)= s(1)+s(3)+s(m)&=1+\frac{1}{3}+\frac{1}{m},\\
S(P,b)= s(2)+s(2)+s(3)&=1+\frac{1}{3},\\
S(P,x)\leq s(4)+s(4)+s(4)&=\frac{3}{4}. 
\end{align*}
Then, for each $x\in X\setminus \{d\}$, $S(P,d)>S(P,x)$. Thus, $f({P}) = d$.
Hence, $P_{-1}$ is worst for $(P_1, f(P_1,\cdot))$. Let $P' _{1} \in \mathcal{P}$ be such that $t_1(P'_1)=c$ and $t_m(P'_1)=d$. Then,
\begin{align*}
S(P'_1, P_{-1},d)=S(P,d)&=1+\frac{1}{2}+\frac{1}{m}\quad \text{and}\\
S(P'_1, P_{-1},c)=S(P,c)+s(1)-s(3)&=2+\frac{1}{m}.
\end{align*}
Thus, $S(P'_1, P_{-1},c)>S(P'_1, P_{-1},d)$ and $f(P'_1, P_{-1})\neq d$. Hence,
\[f(P'_1, P_{-1})\mathrel {P_1} f(P)=d.\]
Therefore, $f$ is not WCSP.\\\\
\textbf{Case 4. $\bm{n=2k+4}$ with $\bm {k\in \mathbb{Z}_{\geq 0}}$:}\\
Let $P \in \mathcal{P}^{2k+4}$ be given by the following table:\footref{fot:table}
\begin{table}[H]
\centering
\begin{tabular}{cccc|ccc|ccc}
$P_1$&$P_2$&$P_3$&$P_4$&$P_{5}$&$\cdots$&$P_{k+4}$&$P_{k+5}$&$\cdots$&$P_{2k+4}$\\\hline
$a$&$c$&$c$&$b$&$c$&$\cdots$&$c$&$b$&$\cdots$&$b$\\
$b$&$b$&$b$&$c$&$b$&$\cdots$&$b$&$c$&$\cdots$&$c$\\
$\vdots$& $\vdots$& $\vdots$& $\vdots$& $a$& $\cdots$& $a$& $a$& $\cdots$&$a$\\
$c$&$a$&$a$&$a$&$\vdots$&$\cdots$&$\vdots$&$\vdots$&$\cdots$&$\vdots$\\
\multicolumn{4}{c}{} & \multicolumn{3}{c}{\upbracefill}& \multicolumn{3}{c}{\upbracefill}\\
\multicolumn{4}{c}{} & \multicolumn{3}{c}{$k$ agents}& \multicolumn{3}{c}{$k$ agents}\\
\end{tabular}
\end{table}
Then, for each $x\in X\setminus\{a,b,c\}$,
\begin{align*}
S(P,c)= s(m)+s(1)+s(1)+s(2)+k\times s(1)+k\times s(2)&=2+\frac{1}{2}+\frac{1}{m}+\frac{3k}{2},\\
S(P,b)= s(2)+s(2)+s(2)+s(1)+k\times s(2)+k\times s(1)&=2+\frac{1}{2}+\frac{3k}{2},\\
S(P,a)= s(1)+s(m)+s(m)+s(m)+2k\times s(3)&= 1+\frac{3}{m}+\frac{2k}{3},\\
S(P,x)\leq s(3)+s(3)+s(3)+s(3)+2k\times s(4)&=1+\frac{1}{3}+\frac{2k}{4}.
\end{align*}
Hence, for each $x\in X\setminus\{c\}$, $S(P,c)>S(P,x)$. Thus, $f({P}) = c$. Hence, $P_{-1}$ is worst for $(P_1, f(P_1,\cdot))$. Let $P' _{1} \in \mathcal{P}$ be such that $t_1(P'_1)=b$ and $t_m(P'_1)=c$. Then,
\begin{align*}
S(P'_1, P_{-1},c)=S(P,c)&=2+\frac{1}{2}+\frac{1}{m}+\frac{3k}{2}\quad\text{and}\\
S(P'_1, P_{-1},b)=S(P,b)+s(1)-s(2)&=3+\frac{3k}{2}.
\end{align*}
Thus, $S(P'_1, P_{-1},b)>S(P'_1, P_{-1},c)$. Then, $f(P'_1, P_{-1})\neq c$.
Hence, \[f(P'_1, P_{-1})\mathrel {P_1} f(P)=c.\]
Therefore, $f$ is not WCSP.\\\\
\textbf{Case 5. $\bm{n=2k+5}$ with $\bm {k\in \mathbb{Z}_{\geq 0}}$:}\\
Let $P \in \mathcal{P}^{2k+5}$ be given by the following table:\footref{fot:table}
\begin{table}[H]
\centering
\begin{tabular}{ccccc|ccc|ccc}
$P_1$&$P_2$&$P_3$ &$P_4$&$P_5$&$P_6$&$\cdots$&$P_{k+5}$&$P_{k+6}$&$\cdots$&$P_{2k+5}$\\\hline
$a$&$c$&$c$ &$b$&$b$&$c$&$\cdots$&$c$&$b$&$\cdots$&$b$\\
$b$&$a$&$a$&$c$ &$c$&$b$&$\cdots$&$b$&$c$&$\cdots$&$c$\\
$\vdots$& $\vdots$& $\vdots$ &$\vdots$& $\vdots$& $a$& $\cdots$& $a$& $a$& $\cdots$&$a$\\
$c$&$b$&$b$&$a$&$a$&$\vdots$&$\cdots$&$\vdots$&$\vdots$&$\cdots$&$\vdots$\\
\multicolumn{5}{c}{} & \multicolumn{3}{c}{\upbracefill}& \multicolumn{3}{c}{\upbracefill}\\
\multicolumn{5}{c}{} & \multicolumn{3}{c}{$k$ agents}& \multicolumn{3}{c}{$k$ agents}\\
\end{tabular}
\end{table}
Then, for each $x\in X\setminus\{a,b,c\}$,
\begin{align*}
S(P,c)= s(m)+s(1)+s(1)+s(2)+s(2)+k\times s(1)+k\times s(2)&=3+\frac{1}{m}+\frac{3k}{2},\\
S(P,b)= s(2)+s(m)+s(m)+s(1)+s(1)+k\times s(2)+k\times s(1)&=2+\frac{1}{2}+\frac{2}{m}+\frac{3k}{2},\\
S(P,a)= s(1)+s(2)+s(2)+s(m)+s(m)+2k\times s(3)&= 2+\frac{2}{m}+\frac{2k}{3},\\
S(P,x)\leq s(3)+s(3)+s(3)+s(3)+s(3)+2k\times s(4)&=1+\frac{2}{3}+\frac{2k}{4}.
\end{align*}
Therefore, for each $x\in X\setminus\{c\}$, $S(P,c)>S(P,x)$. Thus, $f({P}) = c$. Hence, $P_{-1}$ is worst for $(P_1, f(P_1,\cdot))$. Let $P' _{1} \in \mathcal{P}$ be such that $t_1(P'_1)=b$ and $t_m(P'_1)=c$. Then,
\begin{align*}
S(P'_1, P_{-1},c)=S(P,c)&=3+\frac{1}{m}+\frac{3k}{2}\quad\text{and}\\
 S(P'_1, P_{-1},b)=S(P,b)+s(1)-s(2)&=3+\frac{2}{m}+\frac{3k}{2}.
\end{align*}
Thus, $S(P'_1, P_{-1},b)>S(P'_1, P_{-1},c)$. Then, $f(P'_1, P_{-1})\neq c$. Hence,
\[f(P'_1, P_{-1})\mathrel {P_1} f(P)=c.\]
Therefore, $f$ is not WCSP.\\\\
Considering all five cases, we have:
\begin{itemize}
\item If $n=2$, then $f$ is not WCSP (Case 1).
\item If $n=3$, then $f$ is not WCSP (Cases 2 and 3).
\item If $n\ge4$, then $f$ is not WCSP (Cases 4 and 5).
\end{itemize}
Thus, we conclude no Dowdall rule $f$ is WCSP.
\hfill \qed

\subsection{Proof of Theorem \ref{thm:nonapv}\label{app:nonapv}}
We shall show that if $n\geq m$, then no scoring rule with $\succ$, other than anti-plurality rules, is WCSP. Assume that $n\geq m$. Let $f : \mathcal{P}^n \to X$ be a scoring rule with $\succ$ and not an anti-plurality rule. Then, by the definition of a scoring rule, $s(1)\geq \cdots \geq s(m)$ and $s(1)>s(m)$.
Since $f$ is not an anti-plurality rule, $s(1)>s(m-1)$. There are three cases to consider.\\\\
\textbf{Case 1. $\bm{s(1)=s(2)}$ and $\bm{n=m+k}$ with $\bm{k\in \mathbb{Z}_{\geq 0}}$:}\\
Let $P \in \mathcal{P}^{m+k}$ be given by the following table:\footref{fot:table}
\begin{table}[H]
\centering
\begin{tabular}{cccc|ccc}
$P_1$&$P_2$&$\cdots$&$P_m$&$P_{m+1}$&$\cdots$&$P_{m+k}$\\ \hline
$x_m$&$x_1$&$\cdots$&\cellcolor{lightgray}$\vdots$&$x_1$&$\cdots$&$x_1$\\
\cellcolor{lightgray}$\vdots$&$x_m$&$\cdots$&$x_2$&$x_2$&$\cdots$&$x_2$\\
$x_2$&\cellcolor{lightgray}$\vdots$&$\cdots$&$x_1$&$\vdots$ &$\cdots$&$\vdots$\\
$x_1$&$x_2$&$\cdots$&$x_m$&$\vdots$ &$\cdots$&$\vdots$\\
\multicolumn{4}{c}{}& \multicolumn{3}{c}{\upbracefill}\\
\multicolumn{4}{c}{}& \multicolumn{3}{c}{$k$ agents}\\
\end{tabular}
\end{table}
Then, for each $x\in X$, $S(P,x_1) =S(P,x_2)\geq S(P,x)$. By the definition of $\succ$, $f({P}) = x_1 $. Thus, $P_{-1}$ is worst for $(P_1, f(P_1,\cdot))$. Let $P' _{1} \in \mathcal{P}$ be such that $t_1(P'_1)=x_2$ and $t_m(P'_1)=x_1$. Since $s(1)>s(m-1)$, $S(P'_1, P_{-1},x_2)=S(P,x_2)+s(1)-s(m-1)>S(P,x_1)=S(P'_1,P_{-1},x_1)$.
Then, $f(P'_1, P_{-1})\neq x_1$. Thus,
\[f(P'_1, P_{-1})\mathrel {P_1} f(P)=x_1.\]
Therefore, $f$ is not WCSP.\\\\
\textbf{Case 2. $\bm{s(1)>s(2)}$ and $\bm{n=m+2k}$ with $\bm {k\in \mathbb{Z}_{\geq 0}}$:}\\
Let $P \in \mathcal{P}^{m+2k}$ be given by the following table:\footref{fot:table}
\begin{table}[H]
\centering
\begin{tabular}{cccc|ccc|ccc}
$P_1$&$P_2$&$\cdots$&$P_m$&$P_{m+1}$&$\cdots$&$P_{m+k}$&$P_{m+k+1}$&$\cdots$&$P_{m+2k}$\\ \hline
$x_m$&$x_1$&$\cdots$&\cellcolor{lightgray}$\vdots$&$x_2$&$\cdots$&$x_2$&$x_1$ &$\cdots$&$x_1$\\
\cellcolor{lightgray}$\vdots$&$x_m$&$\cdots$&$x_2$&$x_1$&$\cdots$&$x_1$&$x_2$ &$\cdots$&$x_2$\\
$x_2$&\cellcolor{lightgray}$\vdots$&$\cdots$&$x_1$&$\vdots$ &$\cdots$&$\vdots$ &$\vdots$&$\cdots$&$\vdots$\\
$x_1$&$x_2$&$\cdots$&$x_m$&$\vdots$ &$\cdots$&$\vdots$ &$\vdots$&$\cdots$&$\vdots$\\
\multicolumn{4}{c}{}& \multicolumn{3}{c}{\upbracefill}& \multicolumn{3}{c}{\upbracefill}\\
\multicolumn{4}{c}{}& \multicolumn{3}{c}{$k$ agents} & \multicolumn{3}{c}{$k$ agents}\\
\end{tabular}
\end{table}
Then, for each $x\in X$, $S(P,x_1) =S(P,x_2)\geq S(P,x)$. By the definition of $\succ$, $f({P}) = x_1 $. We can prove that $f$ is not WCSP in a manner similar to Case 1.\\\\
\textbf{Case 3. $\bm{s(1)>s(2)}$ and $\bm{n=m+2k+1}$ with $\bm{k\in \mathbb{Z}_{\geq 0}}$:}\\
Let $P \in \mathcal{P}^{m+2k+1}$ be given by the following table:\footref{fot:table}
\begin{table}[H]
\centering
\begin{tabular}{cccc|ccc|ccc}
$P_1$ & $P_2$ &$\cdots$&$P_m$ &$P_{m+1}$&$\cdots$&$P_{m+k}$&$P_{m+k+1}$&$\cdots$&$P_{m+2k+1}$\\\hline
$x_m$& $x_2$& $\cdots$&\cellcolor{lightgray}$\vdots$&$x_1$&$\cdots$&$x_1$&$x_2$&$\cdots$&$x_2$\\
\cellcolor{lightgray}$\vdots$&$x_m$&$\cdots$&$x_1$&$x_2$&$\cdots$&$x_2$&$x_1$&$\cdots$&$x_1$\\
$x_1$&\cellcolor{lightgray}$\vdots$&$\cdots$&$x_2$&$\vdots$&$\cdots$&$\vdots$&$\vdots$&$\cdots$&$\vdots$\\
$x_2$& $x_1$& $\cdots$& $x_m$& $\vdots$& $\cdots$& $\vdots$& $\vdots$& $\cdots$&$\vdots$\\
\multicolumn{4}{c}{}& \multicolumn{3}{c}{\upbracefill}& \multicolumn{3}{c}{\upbracefill}\\
\multicolumn{4}{c}{}& \multicolumn{3}{c}{$k$ agents}& \multicolumn{3}{c}{$k+1$ agents}\\
\end{tabular}
\end{table}
Then, for each $x\in X\setminus \{x_2\}$,
\begin{align*}
S(P,x_2) &>S(P,x_1)\geq S(P,x)\quad \text{and}\\
S(P,x_2) &=S(P,x_1)+s(1)-s(2).
\end{align*} 
Thus, $f({P}) = x_2 $ and $P_{-1}$ is worst for $(P_1, f(P_1,\cdot))$. Let $P' _{i} \in \mathcal{P}$ be such that $t_1(P'_1)=x_1$ and $t_m(P'_1)=x_2$. Since $s(1)-s(m-1)\geq s(1)-s(2)$, $S(P'_1, P_{-1},x_1)=S(P,x_1)+s(1)-s(m-1)\geq S(P,x_1) +s(1)-s(2)=S(P,x_2)=s(P'_1, P_{-1},x_2)$. By the definition of $\succ$, $f(P'_1, P_{-1})\neq x_2$. Thus,
\[f(P'_1, P_{-1})\mathrel{P_1} f(P)=x_2.\]
Therefore, $f$ is not WCSP.\\\\
Considering all three cases, we have:
\begin{itemize}
\item If $s(1)=s(2)$, then $f$ is not WCSP (Case 1).
\item If $s(1)>s(2)$, then $f$ is not WCSP (Cases 2 and 3).
\end{itemize}
Thus, we conclude if $n\geq m$, then no scoring rule $f$ with $\succ$, other than anti-plurality rules, is WCSP.
\hfill \qed
\end{document}